\documentclass[11pt]{article}
\usepackage{amsmath,amsthm,epsfig,amsfonts,color,graphicx,subfigure,float,booktabs,amssymb}
\usepackage{mathrsfs,longtable,geometry}
\usepackage{cite}
\usepackage{url}
\allowdisplaybreaks[4]

\def\singlespace{\def\baselinestretch{1.5}\@normalsize}

\newtheorem{theorem}{{Theorem}}

\newtheorem{remark}{{Remark}}
\renewcommand{\baselinestretch}{1.9}

\usepackage{appendix}
\usepackage{enumerate}
\usepackage{algorithm}  
\usepackage{algpseudocode}  
\usepackage{amsmath}  



\def\marginnote#1{\setbox0=\vtop{\hsize4pc
\small\raggedright\noindent\baselineskip9pt \rightskip=0.5pc plus
1.5pc #1}\leavevmode \vadjust{\dimen0=\dp0
\kern-\ht0\hbox{\kern-4.00pc\box0}\kern-\dimen0}}
\def\lboxit#1{\vbox{\hrule\hbox{\vrule\kern6pt
\vbox{\kern6pt#1\kern6pt}\kern6pt\vrule}\hrule}}

\begin{document}
\thispagestyle{empty}
\begin{center}
{\Large \textbf{Inference for a New Signed Integer Valued Autoregressive  Model Based on Pegram's Operator}}
\end{center}
\vskip 10 pt \centerline{\sc Yinong  Wu$^1$, Dehui Wang$^2$ }
\begin{footnotetext}
{\hspace*{-0.25 in}
$^1$School of Mathematics, Jilin University, 2699 Qianjin Street, Changchun, 130012, Jilin, P.R.China.\\
$^2$School of Mathematics, Liaoning University, 66 Chongshan Middle Road, Shenyang, 110000, Liaoning. P.R.China.\\}
\end{footnotetext}

{\bf Abstract.}
In the current study, a brand-new SINARS(1) model  is proposed for stationary discrete time series defined on $\boldsymbol{Z}$,  based on extended binomial distribution and the Pegram's operator. The model effectively characterizes the series of positive and negative integer values generated after differencing some non-stationary time series. The model's attributes are addressed. For the parameter estimation of the model, the conditional maximum likelihood method and Yule-Walker method are taken into consideration. And we prove the asymptotic normality of CML method. By using these two methods, we simulate our model comparing with some relevant ones proposed before. The model can deal with positive  or negative autocorrelation data.  The analysis of the number of differenced daily new cases in Barbados is done using the suggested model. \\
{\it Key words and phrases}: Asymptotic distribution, Skellam distribution, Extended binomial distribution, Pegram's operator, CMLE.\\

\section{Introduction }
It is well fact that modeling count time series has a crucial role in both financial and medical applications. In some cases, some non-stationary time series    produces series of positive and negative integer values after differenced. Therefore, the modeling of integer value series is extremely important.There are numerous models used to analyze discrete time series defined on $\boldsymbol{Z}$. For instance, Andersson  and Karlis\cite{andersson2014parametric} presented the SINAR process with Skellam innovations (SINARS), which can handle data described on both positive and negative integers. An integer-valued autoregressive process of order $p$ with a signed binomial thinning operator (INARS(p)) was also introduced by Kachour et al.\cite{kachour2011p}, which expanded the order from one to $p$. And as expansion of INARS(p), Zhang et al.\cite{zhang2010inference} proposed INAR(p) process with signed generalized power series thinning operator, which   can dealt with negative integer-valued time series. As supplementary, 
Wang et al.\cite{wang2010generalized} introduced generalized RCINAR(p) process with signed thinning operator.
Furthermore, Chesneau and Kachour \cite{chesneau2012parametric} considered the P-INAR(1) process with Rademacher(p) distribution to deal with data on $\boldsymbol{Z}$.
  
The ZOIPLINAR(1) model, which was developed by Mohammadi et al.\cite{mohammadi2022zero} in 2021, is a stationary INAR(1) model with zero-and-one inflated Poisson-Lindley distributed innovations. 
An integer-valued time series model of order one with an innovation structure of the zero-and-one inflated type is proposed  to describe nonnegative data  with plenty of zero and one. The analysis of two medical series, including the quantity of fresh COVID-19-infected series from Barbados and data on Poliomyelitis in \cite{mohammadi2022modeling}, is then conducted using the zero-and-one-inflated INAR model.

An INAR(1) model based on the combination of Pegram and thinning operators (MPT) with serially dependent innovation was introduced by  Shirozhan and  Mohammadpour \cite{shirozhan2020inar}, which offered additional flexibility in empirical modeling.

Since the PDINAR model is most used processing series with Skellam marginal distribution, we want to extend it to unconstrained marginal distribution.  We integrate the PDINAR model with the Pegram operator based on these studies. The following equation describes the PDINAR(1) time series model:
$
Z_{t}=\delta S_{\alpha,\theta}\left(Z_{t-1}\right)+\varepsilon_{t},
$
where $\delta$ is the sign of autocorrelation, $S_{\alpha,\theta}(Z)$ is extended binomial operator and $\{\varepsilon_t\}$ is a sequence of independent and identically  Skellam distributed random variables with mean $\theta_1-\theta_2$ and variance $\theta_1+\theta_2$.
It is proposed by Alzaid and Omair\cite{alzaid2014poisson} and applied in financial data,  Saudi 
Telecommunication Company (STC) stock and the Electricity stock. 


In this paper, we  give a new signed integer-valued autoregressive process based on the  Pegram's operator. The remainder of the paper is organized as follows.
In this paper, we first show  preparation information.
Section 2 presents definition and fundamental properties for extended binomial distribution, Skellam distribution and properties of the model we proposed. 
Section 3 gives conditional maximum likelihood method (CML) and Yuler-Walker method  for estimating the unknown parameters. It is shown that CML method is consistent and asymptotically normal.
Section 4 reports some simulation results. 
Section 5 gives applications for
the number of new cases in Barbodos.
Section 6 shows conclusion of results and prospect work to do. 
Details of proofs are given in Appendix.

\section{Definition and Basic Properties}
In this part, we first show the definition of Skellam distribution,  afterward  we present essential properties for the two distributions and then for the model. 
\subsection{Skellam Distribution}
As it is defined in \cite{jg1946frequency}, a random variable $Z$ in $\mathbf{Z}$ has Skellam distribution with parameters $\theta_{1} \geq 0$ and $\theta_{2} \geq 0$ if \\
$P(Z=z)=e^{-\theta_{1}-\theta_{2}}\left(\frac{\theta_{1}}{\theta_{2}}\right)^{z / 2} I_{z}\left(2 \sqrt{\theta_{1} \theta_{2}}\right), \quad z=\ldots,-1,0,1, \ldots$\\
where  
\begin{equation}
I_{y}(x)=\left(\frac{x}{2}\right)^{y} \sum_{k=0}^{\infty} \frac{\left(\frac{x^{2}}{4}\right)^{k}}{k !(y+k) !} 
\end{equation}
is the modified Bessel function  of the first kind. For convenience, we use $PD(\theta_1,\theta_2)$ to describe Skellam distribution.
\begin{remark}
Alzaid et al.\cite{alzaid2010poisson} showed another equivalent formular of Skellam distribution
$
P(Z=z)=e^{-\theta_1-\theta_2}\left(\theta_1 \theta_2\right)^{\max \{0,-z\}} \theta_1^{y}{}_{0}\widetilde{F}_{1}\left(|z|+1, \theta_1 \theta_2\right), \quad z=\cdots,-1,0,1, \cdots
$
\end{remark}
\subsection{Extended Binomial Distriution}
To consider the orginal definition of EB distribution first proposed by Alzaid and Omair \cite{alzaid2012extended} we can
let
$W \sim PD(\theta_1,\theta_2)$, be independent of
$R \sim PD(\theta_3,\theta_4)$,
then the sum
$Z = W+R \sim PD(\theta_1+\theta_3,\theta_2+\theta_4)$,\\
then, the conditional distribution, $W|Z\sim Conditional PD(z,\theta_1,\theta_2,\theta_3,\theta_4)$ i.e.,
\begin{align}
\notag & P(W=w | Z=z)= \frac{P(W=w) P(R=z-w)}{P(Z=z)} \\
\notag & \quad \quad \quad  \quad \quad \quad \quad \quad = \frac{\left(\frac{\theta_1}{\theta_1+\theta_3}\right)^w{ }_0^w \widetilde{F}_1\left(; w+1 ; \theta_1 \theta_2\right)\left(\frac{\theta_3}{\theta_1+\theta_3}\right)^{z-w}{ }_0 \widetilde{F}_1\left(; z-w+1 ; \theta_3 \theta_4\right)}{{ }_0 \widetilde{F}_1\left(; z+1 ;\left(\theta_1+\theta_3\right)\left(\theta_2+\theta_4\right)\right)}, \\
& \quad \quad \quad  \quad \quad \quad \quad \quad \quad
w=\ldots,-1,0,1, \ldots
\end{align}
By  the constraint $\theta_1\theta_4 = \theta_2\theta_3$ and the following reparametrization:\\
$z=z$,
$p = \frac{\theta_1}{\theta_1+\theta_3}$,
$q = 1-p$,
$\theta = (\theta_1+\theta_3)(\theta_2+\theta_4)$, we derive the EB distribution.

A random variable $X$ in $\mathbf{Z}$ has extended binomial distribution with parameters $0<p<1$, $q = 1-p$, $\theta>0$ and $z \in \mathbf{Z}$, denoted by $X \sim E B(z, p, \theta)$ if 
\begin{equation}
P(X=x)=\frac{p^{x} q^{z-x}{}_{0}\widetilde{F}_{1}\left(; x+1 ; p^{2} \theta\right){}_{0}\widetilde{F}_{1}\left(; z-x+1 ; q^{2} \theta\right)}{_{0} \widetilde{F}_{1}(; z+1 ; \theta)} \qquad    x=\ldots,-1,0,1, \ldots
\end{equation}
where the regularized hypergeometric function $_{0}\widetilde{F}_{1}$ is defined as:
\begin{equation}
_{0}\widetilde{F}_{1}\left(;y; \theta\right)=\sum_{k=0}^{\infty} \frac{\theta^{k}}{k ! \Gamma(y+k)}.
\end{equation}
For convenience, we use $EB(z,p,\theta)$ to describe extended binomial distribution.

\subsubsection{Extended Binomial Thinning Operator}
The extended binomial thinning operator raised by Alzaid and Omair \cite{alzaid2014poisson} has the following presentation
\begin{equation}
S_{\alpha, \theta}(Z)=(\operatorname{sgn} Z) \sum_{i=1}^{|Z|} Y_{i}+\sum_{i=1}^{W(Z)} B_{i}
\end{equation}
where $Y_{i}$ is a sequence of i.i.d.random variables, independent of $B_{i}$, $Z$ and $W(Z)$, such that $P\left(Y_{i}=1\right)=1-P\left(Y_{i}=0\right)=\alpha$, $\left\{B_{i}\right\}$ is a sequence of i.i.d.random variables independent of $Y_{i}$, $Z$ and $W(Z)$ such that\\
$P\left(B_{i}=1\right)=P\left(B_{i}=-1\right)=\alpha(1-\alpha)$ and $P\left(B_{i}=0\right)=1-2 \alpha(1-\alpha)$,\\
and $W(Z)|Z=z$ is a random variable having Bessel distribution \cite{iliopoulos2003simulation} with parameters $(|z|, \theta)$. 

Since $\sum_{i=1}^{|Z|} Y_{i}|Z=z \sim \operatorname{\emph{binomial}}(|z|, \alpha)$, $\sum_{i=1}^{W(Z)} B_{i} |Z=z$ has the distribution with characteristic function given by 
\begin{equation}
\Phi(t)=\frac{{ }_{0} \widetilde{F}_{1}\left(;|z|+1 ; \theta\left(\alpha (1-\alpha) e^{i t}+\alpha(1-\alpha)e^{-i t}+1-2 \alpha(1-\alpha)\right)\right)}{{ }_{0} \widetilde{F}_{1}(;|z|+1 ; \theta)}.
\end{equation}
And it is clear that $S_{\alpha, \theta}(Z)|Z=z\sim EB(z, \alpha, \theta)$.

\subsection{Models and Properties}

Let us define the stationary MESINAR(1) ${Z_t}$ as
\begin{equation}
Z_{t}=\left(\phi,\delta S_{\alpha,\theta}\left(Z_{t-1}\right)\right)*\left(1-\phi,\varepsilon_{t}\right)
\end{equation}
where $\delta = 1$ or $-1$ decided by the sign of the correlation, 
the Pegram operator '*' mix two independent P and Q with the respective mixing weight of $\phi$ and $1-\phi$, to produce a random variable $Z=(\phi,P)*(1-\phi,Q)$. $\{\varepsilon_t\}$ is a sequence of independent and identically  Skellam distributed random variables with mean $\theta_1-\theta_2$ and variance $\theta_1+\theta_2$.

Since the process is Markovian, the one step transition probability is 
\begin{align}
\notag P(Z_t=z_t|Z_{t-1}=z_{t-1}) & =\phi P\left(\delta S_{\alpha,\theta}(Z_{t-1}) = z_t\right) + (1-\phi)P(\varepsilon_{t} =  z_t) \\
\notag & =\phi P\left( S_{\alpha,\theta}(Z_{t-1}) = \delta z_t\right) + (1-\phi)P(\varepsilon_{t} =  z_t) \\
\notag & = \phi p^{\delta z_{t}}(1-p)^{z_{t-1}-\delta z_{t}}\frac{{}_{0}\widetilde{F}_{1}(; \delta z_{t}+1 ; p^{2}\theta){}_{0}\widetilde{F}_{1}(; z_{t-1}-\delta z_{t}+1 ; (1-p)^{2}\theta)}{{}_{0}\widetilde{F}_{1}(; z_{t-1}+1 ; \theta)} \\ 
& \quad 
+ (1-\phi)e^{-\theta_{1}-\theta_{2}}(\theta_{1}\theta_{2})^{\max(0,-  z_{t})}\theta_{1}^{z_{t}}{}_{0}\widetilde{F}_{1}(; |z_{t-1}|+1 ; \theta_{1}\theta_{2}) 
\end{align}

The conditional mean is
\begin{equation}
E[Z_t|Z_{t-1}]=\phi p \delta Z_{t-1}+(1-\phi)(\theta_1-\theta_2).
\end{equation}

The conditional variance is 
%
\begin{align}
\notag Var[Z_t|Z_{t-1}] &=\phi \left( p (1-p)Z_{t-1}+2 p (1-p) \theta \frac{{ }_0 \widetilde{F}_1(; Z_{t-1}+2 ; \theta)}{{ }_0 \widetilde{F}_1(; Z_{t-1}+1 ; \theta)} \right) +(1-\phi)(\theta_1+\theta_2) \\
& \quad +\phi(1-\phi)(p\delta Z_{t-1}-(\theta_1-\theta_2))^2.
\end{align}

The unconditional mean is
\begin{equation}
E[Z_t]=\frac{(1-\phi)(\theta_1-\theta_2)}{1-\phi p \delta}.
\end{equation}

The unconditional variance is 
\begin{align}
 \notag Var[Z_{t}] &=\frac{p(1-p)\phi(1-\phi)(\theta_{1}-\theta_{2})}{(1-\phi^2 p^2-(1-\phi )p^2)(1-\phi p \delta)} 
 - \frac{2\phi(1-\phi)^3p\delta(\theta_{1}-\theta_{2})^2}{(1-\phi^2 p^2-(1-\phi )p^2)(1-\phi p \delta)} \\
\notag & \quad -\frac{(1-\phi)^{2}(\theta_{1}-\theta_{2})^{2}}{(1-\phi p \delta)^{2}}
 +\frac{2\phi p(1-p)\theta}{(1-\phi^2 p^{2}-(1-\phi)p^2)} E_{Z_{t-1}}[\frac{_{0} \widetilde{F}_{1}\left(;Z_{t-1}+2 ; \theta\right)}{_{0} \widetilde{F}_{1}\left(;Z_{t-1}+1 ; \theta\right)}] \\
 & \quad +\frac{(1-\phi)(\theta_{1}+\theta_{2}+(2-\phi)(\theta_{1}
 -\theta_{2})^{2})}{(1-\phi^2 p^{2}-(1-\phi)p^2)}.      
\end{align}

For any $k$ in $\mathbf{Z}$, the covariance is 
\begin{equation}
\gamma_{k}=Cov\left[Z_{t+k}, Z_{t}\right]=(\phi p \delta)^{k}Var[Z_{t}].
\end{equation}
\section{Estimation}
Next we estimate the unknown parameters by two methods: CML estimation method and Yule-Walker method. All estimates are based on the observed data $(Z_0,Z_1,\ldots,Z_n)$.

\subsection{CML Estimation Method}

\subsubsection{CML Estimation for the Model}
For the sake of simplicity,  let $\beta^2 = \theta$,  $\boldsymbol{\omega} = (\phi,p,\beta,\theta_{1},\theta_{2})^\top$, and  $\boldsymbol{\omega}_{0}$ is the true parameter value of $\boldsymbol{\omega}$.
The conditional likelihood for the model is defined by
\begin{align}
\notag\ L_{n}(\boldsymbol{\omega}) & = P(Z_{0}=z_{0})\prod_{t=1}^{n}P(Z_{t}=z_{t}|Z_{t-1}=z_{t-1})\\ \notag\
& = P(Z_{0}=z_{0})\prod_{t=1}^{n}\phi P(S_{p,\beta^2}(Z_{t-1})= \delta z_{t})+(1-\phi)P(\varepsilon_{t}=z_{t}) \\ \notag\
& = P(Z_{0}=z_{0})\prod_{t=1}^{n} \phi p^{\delta z_{t}}(1-p)^{z_{t-1}-\delta z_{t}}\frac{{}_{0}\widetilde{F}_{1}(; \delta z_{t}+1 ; p^{2}\theta){}_{0}\widetilde{F}_{1}(; z_{t-1}-\delta z_{t}+1 ; (1-p)^{2}\beta^2)}{{}_{0}\widetilde{F}_{1}(; z_{t-1}+1 ; \beta^2)}\\
& \quad  \quad  \quad \quad  \quad \quad  \quad \quad  +  (1-\phi)e^{-\theta_{1}-\theta_{2}}(\theta_{1}\theta_{2})^{\max\{0,-z_{t}\}}\theta_{1}^{z_{t}}{}_{0}\widetilde{F}_{1}(; |z_{t-1}|+1 ; \theta_{1}\theta_{2})
\end{align}
\begin{remark}
We can use the modified Bessel function of the first kind to rewrite the transition probability in order to estimate $\boldsymbol{\omega}$. Recall that ${\beta}^2 = \theta$ $(\theta>0,\beta>0)$, thus
\begin{align}
\notag P(Z_{t}=z_{t}|Z_{t-1}=z_{t-1})
& = \phi P(S_{\alpha,\theta}(Z_{t-1})= \delta z_{t})+(1-\phi)P(\varepsilon_{t}=z_{t}) \\ \notag
& = \phi \frac{I_{\delta z_t}(2 p \beta) I_{z_{t-1}-\delta z_t}(2(1-p) \beta)}{I_{z_{t-1}}(2 \beta)}\\ 
& \quad  +(1-\phi)e^{-\theta_1-\theta_2}\left(\frac{\theta_1}{\theta_2}\right)^{z_t / 2} I_{z_{t}}\left(2 \sqrt{\theta_1 \theta_2}\right).
\end{align}
\end{remark}
Moreover, $l_{n}(\boldsymbol{\omega})$ is the conditonal log-likelihood function and is given by
$$
\frac{\partial l_{n}(\boldsymbol{\omega})}{\partial \phi} = \sum_{t=1}^{n} \frac{\frac{I_{\delta z_{t}}(2 \beta  p) I_{z_{t-1}-{\delta z_{t}}}(2 \beta  (1-p))}{I_{z_{t-1}}(2 \beta )}-e^{-\theta_1-\theta_2} \left(\frac{\theta_1}{\theta_2}\right)^{{z_{t}}/2} I_{z_{t}}\left(2 \sqrt{\theta_1 \theta_2}\right)}{\frac{\phi  I_{\delta z_{t}}(2 \beta  p) I_{{z_{t-1}}-{\delta z_{t}}}(2 \beta  (1-p))}{I_{z_{t-1}}(2 \beta )}+(1-\phi ) e^{-\theta_1-\theta_2} \left(\frac{\theta_1}{\theta_2}\right)^{{z_{t}}/2} I_{z_{t}}\left(2 \sqrt{\theta_1 \theta_2}\right)},
$$
\begin{align*}
\frac{\partial l_{n}(\boldsymbol{\omega})}{\partial p} &=  \sum_{t=1}^{n}
\frac{\frac{\beta  \phi  (I_{{\delta z_{t}}-1}(2 \beta  p)+I_{{\delta z_{t}}+1}(2 \beta  p)) I_{{z_{t-1}}-{\delta z_{t}}}(2 \beta  (1-p))}{I_{z_{t-1}}(2 \beta )}}{\frac{\phi  I_{\delta z_{t}}(2 \beta  p) I_{{z_{t-1}}-{\delta z_{t}}}(2 \beta  (1-p))}{I_{z_{t-1}}(2 \beta )}+(1-\phi ) e^{-\theta_1-\theta_2} \left(\frac{\theta_1}{\theta_2}\right)^{{z_{t}}/2} I_{z_{t}}\left(2 \sqrt{\theta_1 \theta_2}\right)} \\
& \quad \quad -\frac{\frac{\beta  \phi  I_{\delta z_{t}}(2 \beta  p) (I_{{z_{t-1}}-{\delta z_{t}}-1}(2 \beta  (1-p))+I_{{z_{t-1}}-{\delta z_{t}}+1}(2 \beta  (1-p)))}{I_{z_{t-1}}(2 \beta )}}{\frac{\phi  I_{\delta z_{t}}(2 \beta  p) I_{{z_{t-1}}-{\delta z_{t}}}(2 \beta  (1-p))}{I_{z_{t-1}}(2 \beta )}+(1-\phi ) e^{-\theta_1-\theta_2} \left(\frac{\theta_1}{\theta_2}\right)^{{z_{t}}/2} I_{z_{t}}\left(2 \sqrt{\theta_1 \theta_2}\right)},
\end{align*}

\begin{align*}
\frac{\partial l_{n}(\boldsymbol{\omega})}{\partial \beta} &= \sum_{t=1}^{n}
\frac{\frac{(-1) \phi  (I_{{z_{t-1}}-1}(2 \beta )+I_{{z_{t-1}}+1}(2 \beta )) I_{\delta z_{t}}(2 \beta  p) I_{{z_{t-1}}-{\delta z_{t}}}(2 \beta  (1-p))}{I_{z_{t-1}}(2 \beta ){}^2}}{\frac{\phi  I_{\delta z_{t}}(2 \beta  p) I_{{z_{t-1}}-{\delta z_{t}}}(2 \beta  (1-p))}{I_{z_{t-1}}(2 \beta )}+(1-\phi ) e^{-\theta_1-\theta_2} \left(\frac{\theta_1}{\theta_2}\right)^{{z_{t}}/2} I_{z_{t}}\left(2 \sqrt{\theta_1 \theta_2}\right)} \\
& \quad \quad +\frac{ 
\frac{(1-p) \phi  I_{\delta z_{t}}(2 \beta  p) (I_{{z_{t-1}}-{\delta z_{t}}-1}(2 \beta  (1-p))+I_{{z_{t-1}}-{\delta z_{t}}+1}(2 \beta  (1-p)))}{I_{z_{t-1}}(2 \beta )}}
{\frac{\phi  I_{\delta z_{t}}(2 \beta  p) I_{{z_{t-1}}-{\delta z_{t}}}(2 \beta  (1-p))}{I_{z_{t-1}}(2 \beta )}+(1-\phi ) e^{-\theta_1-\theta_2} \left(\frac{\theta_1}{\theta_2}\right)^{{z_{t}}/2} I_{z_{t}}\left(2 \sqrt{\theta_1 \theta_2}\right)} \\
& \quad \quad +\frac{\frac{p \phi  (I_{\delta {z_{t}}-1}(2 \beta  p)+I_{{\delta z_{t}}+1}(2 \beta  p)) I_{{z_{t-1}}-{\delta z_{t}}}(2 \beta  (1-p))}{I_{z_{t-1}}(2 \beta )}}{\frac{\phi  I_{\delta z_{t}}(2 \beta  p) I_{{z_{t-1}}-{\delta z_{t}}}(2 \beta  (1-p))}{I_{z_{t-1}}(2 \beta )}+(1-\phi ) e^{-\theta_1-\theta_2} \left(\frac{\theta_1}{\theta_2}\right)^{{z_{t}}/2} I_{z_{t}}\left(2 \sqrt{\theta_1 \theta_2}\right)},
\end{align*}

\begin{align*}
\frac{\partial l_{n}(\boldsymbol{\omega})}{\partial \theta_1} &= \sum_{t=1}^{n}
\frac{(1-\phi ) \left(-e^{-\theta_1-\theta_2}\right) \left(\frac{\theta_1}{\theta_2}\right)^{{z_{t}}/2} I_{z_{t}}\left(2 \sqrt{\theta_1 \theta_2}\right)}{\frac{\phi  I_{\delta z_{t}}(2 \beta  p) I_{{z_{t-1}}-{\delta z_{t}}}(2 \beta  (1-p))}{I_{z_{t-1}}(2 \beta )}+(1-\phi ) e^{-\theta_1-\theta_2} \left(\frac{\theta_1}{\theta_2}\right)^{{z_{t}}/2} I_{z_{t}}\left(2 \sqrt{\theta_1 \theta_2}\right)} \\
& \quad \quad  +\frac{
\frac{\theta_2 (1-\phi ) e^{-\theta_1-\theta_2} \left(\frac{\theta_1}{\theta_2}\right)^{{z_{t}}/2} \left(I_{{z_{t}}-1}\left(2 \sqrt{\theta_1 \theta_2}\right)+I_{{z_{t}}+1}\left(2 \sqrt{\theta_1 \theta_2}\right)\right)}{2 \sqrt{\theta_1 \theta_2}}}{\frac{\phi  I_{\delta z_{t}}(2 \beta  p) I_{{z_{t-1}}-{\delta z_{t}}}(2 \beta  (1-p))}{I_{z_{t-1}}(2 \beta )}+(1-\phi ) e^{-\theta_1-\theta_2} \left(\frac{\theta_1}{\theta_2}\right)^{{z_{t}}/2} I_{z_{t}}\left(2 \sqrt{\theta_1 \theta_2}\right)}\\
& \quad \quad +\frac{\frac{{z_{t}} (1-\phi ) e^{-\theta_1-\theta_2} \left(\frac{\theta_1}{\theta_2}\right)^{\frac{{z_{t}}}{2}-1} I_{z_{t}}\left(2 \sqrt{\theta_1 \theta_2}\right)}{2 \theta_2}}{\frac{\phi  I_{\delta z_{t}}(2 \beta  p) I_{{z_{t-1}}-{\delta z_{t}}}(2 \beta  (1-p))}{I_{z_{t-1}}(2 \beta )}+(1-\phi ) e^{-\theta_1-\theta_2} \left(\frac{\theta_1}{\theta_2}\right)^{{z_{t}}/2} I_{z_{t}}\left(2 \sqrt{\theta_1 \theta_2}\right)},
\end{align*}

\begin{align*}
\frac{\partial l_{n}(\boldsymbol{\omega})}{\partial \theta_2} &= \sum_{t=1}^{n}
\frac{-\frac{\theta_1 {z_{t}} (1-\phi ) e^{-\theta_1-\theta_2} \left(\frac{\theta_1}{\theta_2}\right)^{\frac{{z_{t}}}{2}-1} I_{z_{t}}\left(2 \sqrt{\theta_1 \theta_2}\right)}{2 \theta_2^2}}{\frac{\phi  I_{\delta z_{t}}(2 \beta  p) I_{{z_{t-1}}-{\delta z_{t}}}(2 \beta  (1-p))}{I_{z_{t-1}}(2 \beta )}+(1-\phi ) e^{-\theta_1-\theta_2} \left(\frac{\theta_1}{\theta_2}\right)^{{z_{t}}/2} I_{z_{t}}\left(2 \sqrt{\theta_1 \theta_2}\right)}\\
& \quad \quad +\frac{
(1-\phi ) \left(-e^{-\theta_1-\theta_2}\right) \left(\frac{\theta_1}{\theta_2}\right)^{{z_{t}}/2} I_{z_{t}}\left(2 \sqrt{\theta_1 \theta_2}\right)}{\frac{\phi  I_{\delta z_{t}}(2 \beta  p) I_{{z_{t-1}}-{\delta z_{t}}}(2 \beta  (1-p))}{I_{z_{t-1}}(2 \beta )}+(1-\phi ) e^{-\theta_1-\theta_2} \left(\frac{\theta_1}{\theta_2}\right)^{{z_{t}}/2} I_{z_{t}}\left(2 \sqrt{\theta_1 \theta_2}\right)} \\
& \quad \quad +\frac{\frac{\theta_1 (1-\phi ) e^{-\theta_1-\theta_2} \left(\frac{\theta_1}{\theta_2}\right)^{{z_{t}}/2} \left(I_{{z_{t}}-1}\left(2 \sqrt{\theta_1 \theta_2}\right)+I_{{z_{t}}+1}\left(2 \sqrt{\theta_1 \theta_2}\right)\right)}{2 \sqrt{\theta_1 \theta_2}}}{\frac{\phi  I_{\delta z_{t}}(2 \beta  p) I_{{z_{t-1}}-{\delta z_{t}}}(2 \beta  (1-p))}{I_{z_{t-1}}(2 \beta )}+(1-\phi ) e^{-\theta_1-\theta_2} \left(\frac{\theta_1}{\theta_2}\right)^{{z_{t}}/2} I_{z_{t}}\left(2 \sqrt{\theta_1 \theta_2}\right)}.
\end{align*}
The conditional maximum likelihood estimators (CMLEs) of the unknown parameters $(\phi,p,\beta,\theta_1,\theta_2)^\top$
are obtained numerically by maximizing the conditional log-likelihood
function or by solving the equations related to the score functions. 
\begin{remark}
The relationship between $I_y(\theta)$ and ${ }_0 \tilde{F}_1$ is
\begin{equation}
I_y(\theta)=\left(\frac{\theta}{2}\right)^y{ }_0 \tilde{F}_1\left(; y+1 ; \frac{\theta^2}{4}\right).
\end{equation}
Let $\frac{\theta^2}{4}=m$ and we have
\begin{equation}
{ }_0 \tilde{F}_1(; y+1 ; m)=m^{-y / 2} I_y(2 \sqrt{m}).
\end{equation}
\end{remark}

The following theorem gives the asymptotic properties of CML.

We denote the transition probability $P\left(Z_{t}=z_{t}| Z_{t-1}=z_{t-1}\right)$ as $f(z_{t-1},z_{t},\boldsymbol{\omega}).$
We want to apply the results of Theorem 2.1 of Billingsley and Patrick \cite{billingsley1961statistical} on estimates for the unknown parameters. We just need to verify that the  conditions in Billingsley and Patrick \cite{billingsley1961statistical} are satisfied one by one. Equivalent conditions  are explained in detail in Franke and  Seligmann\cite{franke1993conditional}, and we can also impose.  \\
Let 
$P(m,n) = P(Z_t = n|Z_{t-1} = m) = \phi Z(m,n) + (1-\phi) q(m,n)$,\\
where $Z(m,n) = \frac{I_{\delta n}(2 p \beta) I_{m-\delta n}(2(1-p) \beta)}{I_{m}(2 \beta)}$ and $q(m,n) = e^{-\theta_1-\theta_2}\left(\frac{\theta_1}{\theta_2}\right)^{n / 2} I_{n}\left(2 \sqrt{\theta_1 \theta_2}\right)$.
%
\\ Based on \cite{franke1993conditional}, \cite{miletic2018inar} and \cite{shirozhan2020inar} and we need to verify that the following conditions (C1)-(C7) are met:\\
(C1) For any $z_{t-1}$, the set of $z_{t}$ for which $f(z_{t-1},z_{t},\boldsymbol{\omega}) > 0$ does not depend on $\boldsymbol{\omega}$.\\
(C2) 
$
E\left[\varepsilon_t^3\right]=\sum_{n=0}^{\infty} n^3 q(m, n)<\infty, E\left[Y_t^3\right]=\sum_{n=0}^m n^3 Z(m, n)<\infty
$\\
(C3) For any $z_{t-1}$ and $z_{t}$, $f_{u}(z_{t-1},z_{t},\boldsymbol{\omega})$, $f_{uv}(z_{t-1},z_{t},\boldsymbol{\omega})$ and $f_{uvw}(z_{t-1},z_{t},\boldsymbol{\omega})$ exist and are continuous throughout $\boldsymbol{\Omega}$. 
Then for any $z_{t-1}$,$$g(z_{t-1},z_{t},\boldsymbol{\omega}) = \log f(z_{t-1},z_{t},\boldsymbol{\omega})$$ is almost surely well defined, and 
 $g_{u}(z_{t-1},z_{t},\boldsymbol{\omega})$, $g_{uv}(z_{t-1},z_{t},\boldsymbol{\omega})$ and $g_{uvw}(z_{t-1},z_{t},\boldsymbol{\omega})$ are continuous in $\boldsymbol{\Omega}$.\\
(C4) For any $\boldsymbol{\omega} \in  \boldsymbol{\Omega}$ there exists a neighborhood $N$ of $\boldsymbol{\omega}$ such that for any $u,v,w,z_{t-1}$,
$$
\sum_{z_t \in \boldsymbol{Z}} \sup _{\boldsymbol{\omega}^{\prime} \in N}\left|f_u\left(z_{t-1}, z_t ; \boldsymbol{\omega}^{\prime}\right)\right| \lambda(d z_t)<\infty
$$
$$
\sum_{z_t \in \boldsymbol{Z}} \sup _{\boldsymbol{\omega}^{\prime} \in N}\left|f_{u v}\left(z_{t-1}, z_t ; \boldsymbol{\omega}^{\prime}\right)\right| \lambda(d z_t)<\infty
$$
$$
E_{\boldsymbol{\omega}}[\sup _{\boldsymbol{\omega}^{\prime} \in N} \left|g_{u v w}\left(z_{t-1}, z_t; \boldsymbol{\omega}^{\prime}\right) \right| ]<\infty.
$$
where $\lambda$ is Lebesgue measure.\\
For any $\eta^{'} \in B$; there exists a neighborhood $U$ of  $\eta^{'}$ and for any $\rho^{'} \in  B$; there exists a neighborhood $H$ of $\rho^{'}$ such that:
$$
\sum_{n=-\infty}^{\infty} \sup _{\eta \in U} q(m, n)<\infty, \sum_{n=-\infty}^{\infty} \sup _{\rho \in H} Z(m, n)<\infty ,
$$
and
$$
\sum_{n=-\infty}^{\infty} \sup _{\eta \in U}\left|q_u(m, n)\right|<\infty, \sum_{n=-\infty}^{\infty} \sup _{\eta \in U}\left|q_{u v}(m, n)\right|<\infty.
$$
(C5)  For any $\eta \in B$; there exists a neighborhood $U$ of $\eta^{'}$ and increasing sequences (depending on $\eta^{'}$ and $U$) such that for all non vanishing
$$
\begin{aligned}
\left|q_u(m, n)\right| & \leq \psi_u(m, n) q(m, n) \\
\left|q_{u v}(m, n)\right| & \leq \psi_{u v}(m, n) q(m, n) \\
\left|q_{u v w}(m, n)\right| & \leq \psi_{u v w}(m, n) q(m, n)
\end{aligned}
$$
and with respect to the stationary distribution of the INAR(l) process $\{Z_t\}$
$$
\begin{array}{rr}
E\left[\psi_u^3\left(Z_1, Z_2\right)\right]<\infty, & E\left[Z_1 \psi_{u v}\left(Z_1, Z_2\right)\right]<\infty, \\
E\left[\psi_u\left(Z_1, Z_2\right) \psi_{v w}\left(Z_1, Z_2\right)\right]<\infty, & E\left[\psi_{u v w}\left(Z_1, Z_2\right)\right]<\infty.
\end{array}
$$
(C6) If $(\sigma_{u v }(\boldsymbol{\omega}))_{u,v=1,...5}$ is defined by
$$
\sigma_{u v }(\boldsymbol{\omega}) = E_{\boldsymbol{\omega}}[g_{u}\left(z_{t-1}, z_t; \boldsymbol{\omega}\right)g_{v}\left(z_{t-1}, z_t; \boldsymbol{\omega}\right) ] \quad \quad u,v=1,...5
$$
then the fisher information matrix $\sigma(\boldsymbol{\omega}) = \sigma_{u v }(\boldsymbol{\omega})$ is
nonsingular.\\
(C7) For each $\boldsymbol{\omega} \in \boldsymbol{\Omega}$, the stationary distribution, which by assumption exists is unique, has the property that for each $\xi \in Z$, $p_{\boldsymbol{\omega}}(\xi,\cdot)$ is absolutely continuous with respect to $p_{\boldsymbol{\omega}}(\cdot)$.

\begin{theorem}
Suppose that $[z_{n},L_n(;\boldsymbol{\omega}),\boldsymbol{\Omega}]$ satisfies conditions (C1)-(C7),
that $\boldsymbol{\omega}_0$ is  the true value of the parameter and that $\hat{\boldsymbol{\omega}}_{n}$ is a consistent solution of the maximum-likelihood equations. 
Then as $n\to\infty$,

(a) $\hat{\boldsymbol{\omega}}_{n} \to \boldsymbol{\omega}_{0}$ almost surely,

(b) $\sqrt{n}(\hat{\boldsymbol{\omega}}_{n}-\boldsymbol{\omega}_{0}) \stackrel{d}{\longrightarrow} N\left(0, G^{-1}\right)$, 
where $G=(\sigma_{uv}),\\ 
\sigma_{u v}=-E_{\boldsymbol{\omega}}\left[\frac{\partial l_{n}(\boldsymbol{\omega}) }{\partial \omega_{u} \omega_{v}}\right]=E_{\boldsymbol{\omega}}\left[\frac{\partial l_{n}(\boldsymbol{\omega})}{\partial \omega_{u}} \cdot \frac{\partial l_{n}(\boldsymbol{\omega}) }{\partial \omega_{v}}\right],1\leq u,v \leq 5$.

\end{theorem}
\begin{proof}
See Appendix A for details.
\end{proof}

\subsection{Yule-Walker Estimation}
To apply the modified Yule-Walker method applied by Mileti{\'c} Ili{\'c} et al. \cite{miletic2018inar} and  Shirozhan and  Mohammadpour\cite{shirozhan2020inar}, we have estimation process as follows. The following set of equations are for estimating $\phi$, $p$, $\theta$, $\theta_1$ and $\theta_2$:
\begin{equation}
\hat{\phi p}=\hat{\delta} r_1=\hat{\delta} \frac{\sum_{t=1}^{n}\left(z_t-\bar{z}\right)\left(z_{t+1}-\bar{z}\right)}{\sum_{t=1}^{n+1}\left(z_t-\bar{z}\right)^2}, \quad
S_{z}^{2} = \frac{1}{n} \sum_{t=1}^{n+1}\left(z_t-\bar{z}\right)^2,
\end{equation}
where  $\bar{z}$ is the sample mean,  ${S_{z}^{2}}$ is the sample variance, and $r_{1}$ is the sample autocorrelation.
\begin{equation}
E[Z_{t}]=\frac{(1-\phi)(\theta_1-\theta_2)}{(1-\phi  p \hat{\delta})}=\bar{z}.
\end{equation}
The unconditional variance is 
\begin{align}
 \notag Var[Z_{t}] &=\frac{p(1-p)\phi(1-\phi)(\theta_{1}-\theta_{2})}{(1-\phi^2 p^2-(1-\phi )p^2)(1-\phi p \delta)} 
 - \frac{2\phi(1-\phi)^3p\delta(\theta_{1}-\theta_{2})^2}{(1-\phi^2 p^2-(1-\phi )p^2)(1-\phi p \delta)} \\
\notag & \quad -\frac{(1-\phi)^{2}(\theta_{1}-\theta_{2})^{2}}{(1-\phi p \delta)^{2}}
 +\frac{2\phi p(1-p)\theta}{(1-\phi^2 p^{2}-(1-\phi)p^2)} E_{Z_{t-1}}[\frac{_{0} \widetilde{F}_{1}\left(;Z_{t-1}+2 ; \theta\right)}{_{0} \widetilde{F}_{1}\left(;Z_{t-1}+1 ; \theta\right)}] \\
 & \quad +\frac{(1-\phi)(\theta_{1}+\theta_{2}+(2-\phi)(\theta_{1}
 -\theta_{2})^{2})}{(1-\phi^2 p^{2}-(1-\phi)p^2)},    
\end{align}
%
where we substitute $p$ and $\theta$ with  $\hat{p}_{ml}$ and  $\hat{\theta}_{ml}$, i.e., the ML estimators of the parameters $p$ and $\theta$, as the initial estimate of the parameters. And we use the following equation:
\begin{equation}
E_{Z_{t-1}}[\frac{_{0} \widetilde{F}_{1}\left(;Z_{t-1}+2 ; \theta\right)}{_{0} \widetilde{F}_{1}\left(;Z_{t-1}+1 ; \theta\right)}] = 
\frac{1}{n}\sum_{t=2}^{n}\frac{_{0} \widetilde{F}_{1}\left(;Z_{t-1}+2 ; \theta\right)}{_{0} \widetilde{F}_{1}\left(;Z_{t-1}+1 ; \theta\right)}. 
\end{equation}
then we have three equations to find solutions for unknown parameters $\phi$, $\theta_1$ and $\theta_2$, that is, we can apply modified Yule-Walker method.
\section{Simulation}
In this section,  we first outline the process of generating data of the model, then give the parameter settings for five groups, and show the performance of CML and Yule-Walker method.
\subsection{Steps of generation of the model}
All steps of generation of the model are as follows:
\begin{enumerate}
	\item Generate $n$ independent observations from $PD(\theta_{1},\theta_{2})$  which serves as $\left\{\varepsilon_{t}\right\}_{t=1}^{n}$.
	\item Generate $n$ independent observations from $binomial(1,\phi)$ which serves as $\left\{V_t\right\}_{t=1}^n$.
	\item Choose an initial value from $PD(\theta_{1},\theta_{2})$ for  $z_{1}$.
	\item Generate an independent observation from $EB(z_{1},\alpha,\theta)$ and call it $S_{p,\theta}(z_{1}).$ 
	\item Let $z_{2}=V_2 S_{p,\theta}(z_{1})+\left(1-V_2\right)\varepsilon_{2}.$
	\item Generate an independent observation from $EB(z_{2},\alpha,\theta)$  and call it $S_{p,\theta}(z_{2})$ .
	\item Let $z_{3}=V_3S_{p,\theta}(z_{2})+\left(1-V_3\right)\varepsilon_{3}$.
	\item Repeat steps 5 and 6 until we obtain the $n^{th}$ observation.
\end{enumerate}

\subsection{Parameter Settings}
To report the performances of the CML method and Yule-Walker method, we conduct simulation studies under the following settings, and note the four groups of parameters by $\boldsymbol{\omega}_1$, $\boldsymbol{\omega}_2$, $\boldsymbol{\omega}_3$, $\boldsymbol{\omega}_4$. Both the case the sequence has positive correlation $(\delta = 1)$ and negative correlation $(\delta =- 1)$ are considered.
\begin{enumerate}
 \item $(\phi, p,\beta,\theta_{1},\theta_{2})^\top = (0.8,0.5,\sqrt{5},10,10)^\top = \boldsymbol{\omega}_1$, $\delta = 1$,
 \item $(\phi, p,\beta,\theta_{1},\theta_{2})^\top = (0.2, 0.4, 2, 9, 7)^\top = \boldsymbol{\omega}_2$, $\delta = 1$,
\item $(\phi, p,\beta,\theta_{1},\theta_{2})^\top = (0.2,0.4,\sqrt{5},5,5)^\top = \boldsymbol{\omega}_3$, $\delta = -1$, 
\item $(\phi, p,\beta,\theta_{1},\theta_{2})^\top = (0.2,0.8,\sqrt{5},10,10)^\top = \boldsymbol{\omega}_4$, $\delta = -1$,
\end{enumerate}
\begin{figure} 
\subfigure[Sample path of $\boldsymbol{\omega}_1$ .] {
 \label{fig:a}     
\includegraphics[width=0.5\columnwidth,height = 9cm]{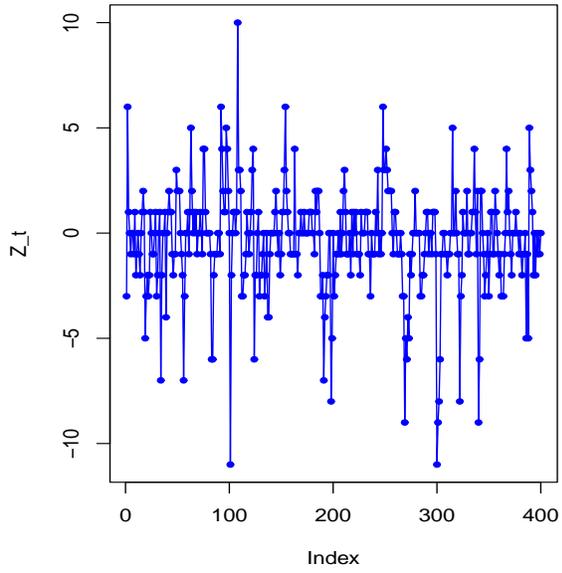} 
}     
\subfigure[Sample path of $\boldsymbol{\omega}_2$.] { 
\label{fig:b}     
\includegraphics[width=0.5\columnwidth,height = 9cm]{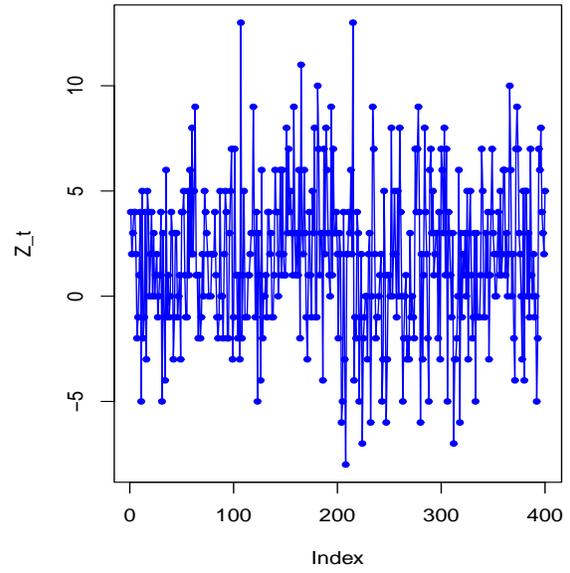}     
}    \\
\subfigure[Sample path of $\boldsymbol{\omega}_3$.] { 
\label{fig:c}     
\includegraphics[width=0.5\columnwidth,height = 9cm]{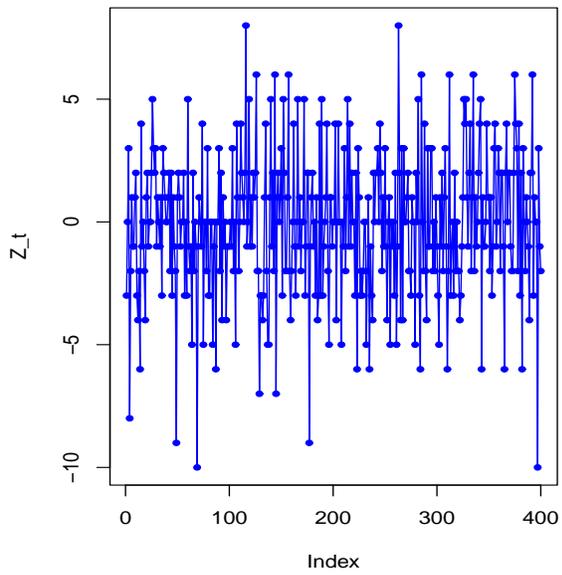}     
}    
\subfigure[Sample path of $\boldsymbol{\omega}_4$.] { 
\label{fig:d}     
\includegraphics[width=0.5\columnwidth,height = 9cm]{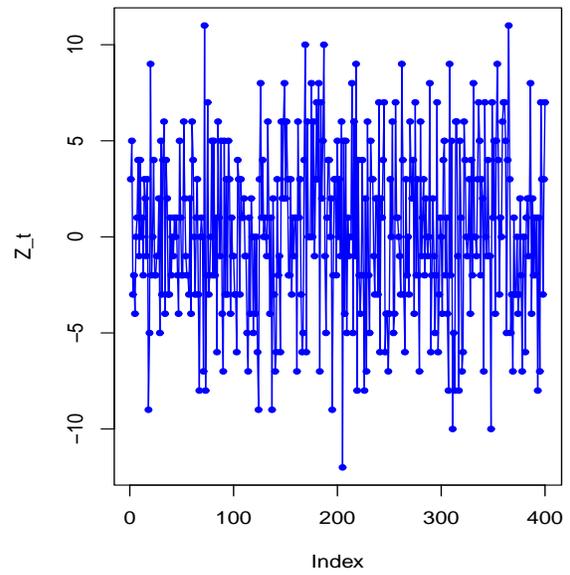}     
}   
\caption{ Sample paths of the four parameter groups. }     
\label{fig1}    
\end{figure}
\subsection{Some Numerical Results}
\subsubsection{Results for CML and Yule-Walker}
The performance of the estimators is checked by a small Monte Carlo simulation using different sample sizes ($n$ = 200, 400, 4000). All the simulations in this study are performed
under the $R$ software based on $N = 100$ replications.
Table 2 gives mean and MSE for the conditional maximum likelihood and Yule-Walker method. Based on the table, we find that the estimates are convergent to their values. Within each case, parameter $\phi$ takes small value 0.2, and large value 0.8, 0.9. Moreover, we choose the case $\theta_1 = \theta_2$ and the case $\theta_1 \not= \theta_2$. Since the value of $\delta$ can be $1$ or $-1$, we also consider the estimation of the negative correlation cases. Table 1 and Table 2 show that obtained estimaters  are convergent to their values in all the cases. Also, increasing the sample size implies smaller MSE. Moreover, we can notice that Yule-Walker estimator gives larger MSE in almost every case. All of the above, we can 
conclude that in most cases, the CMLEs provide good performance which was expected with respect to all parameters. 
\begin{table}[htbp] 
\centering
\label{T2}\caption{Descriptive statistics of daily new cases in Barbados with lag one difference.} \vspace{0.1in}
\begin{tabular}{cccccccc}
\hline
Variable          & no. & Mean       & Variance & Minimum & Median & Maximum & Range \\ \hline
Barbados. cases          & 292              & 1.3527  &5.6037   &0 & 0 &16 &16                     \\ 
Barbados. difference   & 291               &-0.0068   &8.4879 &-14 & 0&  14&  28    \\ \hline

\end{tabular}
\end{table}
\begingroup
\setlength{\tabcolsep}{5pt} 
\renewcommand{\arraystretch}{0.6} 
\begin{table}[htbp]
\label{T2}\caption{Mean value and mean square error (MSE) are reported for $mymodel^{+}$ and $mymodel^{-}$ through CML and Yule-walker method for $(\phi,p,\beta,\theta_1,\theta_2)^\top$
and with $n=200,400$, $800$ and $4000$.} \vspace{0.1in}
\centering
\begin{tabular}{cccccccccc}
\hline
$N $     & $\phi_{ml} $    & $p_{ml}$     & $\beta_{ml}$     & $\theta_{1ml}$ & $\theta_{2ml}$ & $\phi_{yw}$   & $\theta_{1yw}$ & $\theta_{2yw}$ & $\theta_{1yw}-\theta_{2yw}$  \\  \hline
\multicolumn{10}{c}{(i) True values $ \boldsymbol{\omega}_1= (0.8, 0.5, \sqrt{5}, 10, 10)^\top, \delta = 1$}                                                   \\
200  & 0.8004 & 0.5009 & 2.2356  & 10.5380  & 10.3026  & 0.7908 & 15.3333   & 15.3913   & -0.0580           \\
MSE  & 0.0019 & 0.0024 & 0.1454  & 9.7403   & 9.7399   & 0.0185 & 332.3275  & 325.2845  & 1.5563            \\
400  & 0.7952 & 0.4998 & 2.2161  & 10.1775  & 10.1759  & 0.7936 & 19.0681   & 19.2959   & -0.2278           \\
MSE  & 0.0014 & 0.0010 & 0.0662  & 4.0363   & 4.2060   & 0.0104 & 4347.6635 & 4729.4396 & 9.3705            \\
800  & 0.7962 & 0.5006 & 2.2319  & 10.0079  & 10.0324  & 0.7893 & 12.0170   & 12.0429   & -0.0258           \\
MSE  & 0.0004 & 0.0006 & 0.0306  & 2.2122   & 1.9988   & 0.0050 & 32.7616   & 33.0771   & 0.2485            \\
4000 & 0.8002 & 0.4996 & 2.2441  & 10.0126  & 10.0012  & 0.7938 & 11.0958   & 11.0919   & 0.0039            \\
MSE  & 0.0001 & 0.0001 & 0.0079  & 0.3718   & 0.3577   & 0.0013 & 6.2700    & 6.2717    & 0.0288            \\
\multicolumn{10}{c}{(ii) True values $\boldsymbol{\omega}_2 = (0.2, 0.4, 2, 9, 7)^\top, \delta = 1$}                                                                          \\
200  & 0.2304 & 0.4192 & 1.9630  & 9.1064   & 7.0075   & 0.1872 & 9.2269    & 7.1646    & 2.0622            \\
MSE  & 0.0063 & 0.0128 & 3.7185  & 0.9755   & 0.9086   & 0.0397 & 8.6072    & 6.1845    & 0.2958            \\
400  & 0.2125 & 0.3956 & 1.7688  & 9.0346   & 7.0050   & 0.2091 & 8.9010    & 6.8413    & 2.0597            \\
MSE  & 0.0026 & 0.0044 & 2.5326  & 0.5209   & 0.4666   & 0.0152 & 1.8819    & 1.3431    & 0.1247            \\
800  & 0.2043 & 0.4060 & 1.4135  & 9.0057   & 6.9987   & 0.2061 & 8.7288    & 6.6983    & 2.0304            \\
MSE  & 0.0011 & 0.0027 & 0.3711  & 0.2784   & 0.2365   & 0.0112 & 1.2221    & 0.9268    & 0.0535            \\
4000 & 0.1995 & 0.3994 & 1.4436  & 9.0136   & 7.0057   & 0.2058 & 8.5688    & 6.5613    & 2.0075            \\
MSE  & 0.0002 & 0.0004 & 0.0836  & 0.0700   & 0.0603   & 0.0020 & 0.3844    & 0.3418    & 0.0110            \\
\multicolumn{10}{c}{(iv) True values $\boldsymbol{\omega}_3 = (0.2, 0.4, \sqrt{5}, 5, 5)^\top, \delta = -1$}                                                    \\
200  & 0.2354 & 0.4093 & 3.6901  & 5.2991   & 5.2297   & 0.1782  & 4.9700           &  5.0335         &   -0.0635                \\
MSE  & 0.0196 & 0.0326 & 26.8456 & 2.0729   & 1.9893   & 0.0372       &          2.5206 & 2.6660          &   0.0768                \\
400  & 0.2127 & 0.4028 & 2.7968  & 5.0779   & 5.0616   & 0.1893       &         4.9007  &4.8959           & 0.0047                  \\
MSE  & 0.0091 & 0.0120 & 7.9427  & 0.3761   & 0.3408   & 0.0191       &         0.8230  &   0.8166        &   0.0405                \\
800  & 0.2227 & 0.3805 & 2.8852  & 5.1111   & 5.0912   & 0.1979       &          4.8497 & 4.8575          &    -0.0078               \\
MSE  & 0.0067 & 0.0058 & 4.9466  & 0.2222   & 0.2192   & 0.0078       &          0.3244 &  0.3160         &   0.0175                \\
4000 & 0.2026 & 0.4002 & 2.3019  & 4.9971   & 5.0057   &  0.2021      &          4.8513 &   4.8653        &  -0.0140                 \\
MSE  & 0.0005 & 0.0009 & 0.2648  & 0.0254   & 0.0258   &  0.0020      &          0.0894 &  0.0881         &    0.0039               \\
\multicolumn{10}{c}{(v) True values $\boldsymbol{\omega}_4 = (0.2, 0.8, \sqrt{5}, 10, 10)^\top, \delta = -1$}                                                  \\
200  & 0.2150 & 0.7982 & 2.7313  & 9.9927   & 10.0072  & 0.1984       &          9.8042 & 9.7416           &    0.0626               \\
MSE  & 0.0047 & 0.0062 & 7.8927  & 1.5760   & 1.4650   &  0.0088      &          1.7339 &  1.8601         &  0.1320                 \\
400  & 0.2067 & 0.7956 & 2.7468  & 9.9741   & 9.9550   &  0.2014      &          9.8563 &  9.8964         &    -0.0400               \\
MSE  & 0.0012 & 0.0027 & 3.0378  & 0.5707   & 0.6029   &  0.0045      &          0.9554 &  0.9539         &   0.0788                \\
800  & 0.2028 & 0.8006 & 2.3352  & 9.9443   & 9.9523   &  0.2010      &          9,8914 &   9.8812        &  0.0103                 \\
MSE  & 0.0008 & 0.0011 & 0.7012  & 0.3170   & 0.3108   &  0.0024      &          0.4759 &   0.4933        &  0.0391                 \\
4000 & 0.2018 & 0.8007 & 2.3115  & 10.0136  & 10.0214  & 0.2007       &          9.8885 &   9.8877        &   0.0008                \\
MSE  & 0.0001 & 0.0002 & 0.1202  & 0.0654   & 0.0629   & 0.0005       &          0.1108 &  0.1094         &  0.0098                 \\ \hline
\end{tabular}
\end{table}
\section{Real Data Example}

\begin{table}[htbp]
\centering
\label{T3}\caption{Estimated parameters through several models, AIC, BIC for the new cases in barbodos from to} \vspace{0.1in}
\begin{tabular}{ccccccccc}
\hline

Model& CML estimates &AIC & BIC & HQIC &Log-likelihood\\
\hline

$\text{PDINAR}^{-}(1)$
&$\hat{\alpha} = -0.4187   $
&1365.8327 &1376.8629 &1365.0418 &-679.9163 \\

&$\hat{\theta}_1 = 1.7735 $
&&&&&\\

&$\hat{\theta}_2 = 1.8309$
&&&&&\\

$\text{MESINAR}^{-}(1)$

&$\hat{\phi} = 0.5680$
&1042.2406 &1060.6243 &1040.9225 &-516.1203   \\

&$\hat{p} = -0.0055$
&&&&& \\

&$\hat{\beta} = 17.1567$
&&&&&\\

&$\hat{\theta}_1 = 3.8991$
&&&&&\\

&$\hat{\theta}_2 = 1.0149$
&&&&&\\

$\text{MSINARS}^{-}(1)$

&$\hat{\phi} = 0.4752$
&1229.3335 &1244.0405 &1228.2790 &-610.6668 \\

&$\hat{\alpha} = -0.6932$
&&&&&\\

&$\hat{\theta}_1 = 5.3686$
&&&&&\\

&$\hat{\theta}_2 = 4.7866$
&&&&&\\

\hline
\end{tabular}
\end{table}
In this section,we apply our model to a real count data obtained from the website \cite{TP-toolbox-web}. The series is about daily new cases in Barbados. We aim to study the change of incrememental of the data.
The series represents the daily new cases of in Barbados from March, 17th 2020 to January, 2nd 2021.

The sample paths, the difference of the sample paths, autocorrelation functions (ACFs) and partial autocorrelation functions (PACFs) of two series are displayed in Figure 2.
The figures suggest that the first order autoregressice models are appropriate for analysing the given data series.

Some   series after differencing include negative values
thus, the advantage of the model we proposed is to fit integer valued time series with possible negative values 
and either positive or negative correlation.

\begin{figure}  
\subfigure[Sample path of daily new cases.] {
 \label{fig:e}     
\includegraphics[width=0.5\columnwidth,height = 9cm]{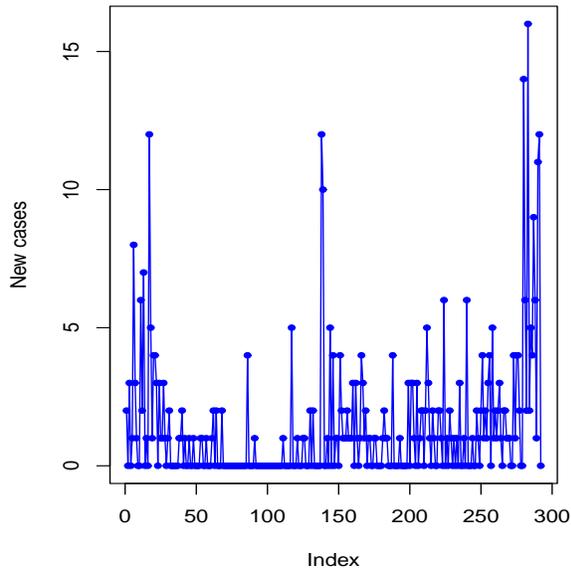}  
} \vspace{-0.1in}
\subfigure[Sample path of differenced daily new cases.] { 
\label{fig:f}     
\includegraphics[width=0.5\columnwidth,height = 9cm]{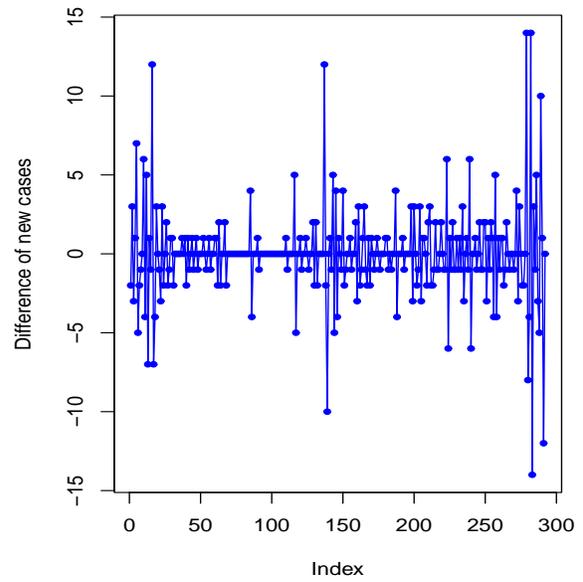}     
}    \\ 
\subfigure[ACF of differenced daily new cases.] { 
\label{fig:g}     
\includegraphics[width=0.5\columnwidth,height = 9cm]{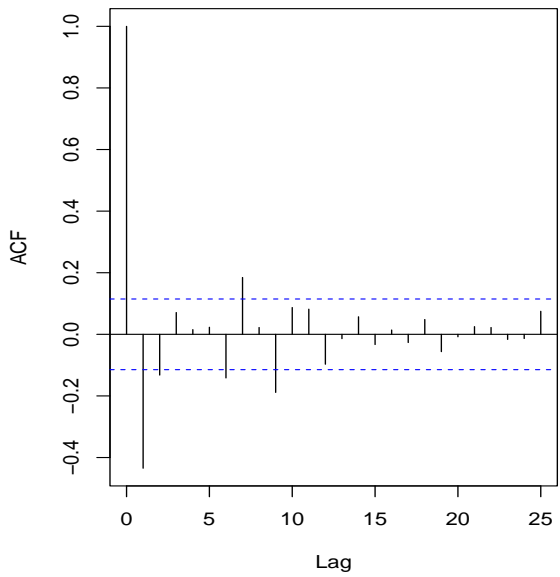}     
}    
\subfigure[PACF of differenced daily new cases.] { 
\label{fig:h}     
\includegraphics[width=0.5\columnwidth,height = 9cm]{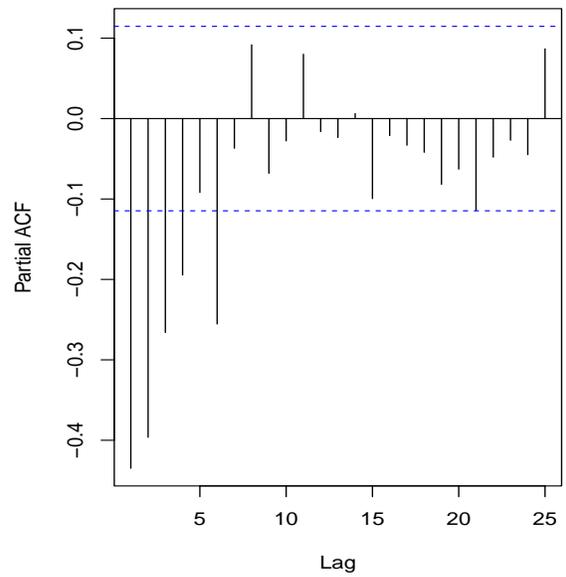}     
}   
\caption{  Sample paths, ACF and PACF of differenced daily new cases. }     
\label{fig3}    
\end{figure}

The results are given in Table 1 and 3.
Table 1 shows descriptive statistics of differenced daily new cases in Barbados with their lag one difference. According to the ACF of differenced daily new cases, the autocorrelation is negative, i.e., $\hat{\delta} = -1$. For the purpose of testing our model compared with other simple distributions or  INAR(1) models, which can analyse positive and negative data, we apply some goodness-of-fit statistics criteria: AIC, BIC, HQIC and log-likelihood. For the purpose of testing our model against some other relevant INAR(1) models, such as $\text{PDINAR}(1)$ and the process $Z_t = U_t\left(\alpha \odot Z_{t-1}\right) +(1-U_t) \varepsilon_t$, where $U_t$ is a sequence of independent and identically Bernoulli random variables with parameter $\phi$, $ "\odot" $ is the operator proposed by Hee-young and Yousung\cite{kim2008non}, and $\left\{ \varepsilon_t\right\}$ is a sequence of independent and identically Poisson difference  distributed random variables with parameters $\theta_1,\theta_2$.

Table 3 presents  some goodness-of-fit statistics for the series. Better model is characterized by smaller values of these statistics.
As it can be seen from these tables, the five statistics  are small for the model we proposed. And the difference $\hat{\theta}_1-\hat{\theta}_2 = 2.8842$ shows that the innovation of original series  is possibly not a single unchanged Poisson distribution.  
Therefore, we can conclude that the  model works well for the real data series.
\section{Conclusion and Prospect}
In this paper, we construct a model based on extended binomial distribution and Skellam distribution with the Pegram's operator dealing with positive  or negative autocorrelation data defined on $\mathbf{Z}$. The best method of all is CML method  via numerical simulation and it performs well on the number of differenced daily new cases in Barbados. 

Next we will extend our model naturally to a SINARS regression model:
$$
\begin{aligned}
Z_t & =(\phi,\delta S_{\alpha,\theta}(Z_{t-1}))*(1-\phi,R_t), t=1,2, \ldots \\
R_t & \sim \operatorname{Skellam}\left(\theta_{1t}, \theta_2\right) \\
\theta_{1t} & =\beta_0+\beta_1 X_{1 t}+\ldots+\beta_k X_{k t}
\end{aligned}
$$
where $X_{it}$ are covariates and the $\beta$’s the associated coefficients.

Another extension is that we consider a mixture of  Bivariate Skellam distribution \cite{omair2022bivariate} and Bivariate extended binomial distribution to deal with binary data.

{\bf Acknowledgments.~}
Our research was supported by Jilin University and Liaoning University.

\nocite{*} 
\bibliographystyle{plain}
\bibliography{ref}

\begin{thebibliography}{10}

\bibitem{TP-toolbox-web}
\url{https://ourworldindata.org/ covid-cases}.

\bibitem{al1987first}
Mohamed~A Al-Osh and Aus~A Alzaid.
\newblock First-order integer-valued autoregressive ({INAR} (1)) process.
\newblock {\em Journal of Time Series Analysis}, 8(3):261--275, 1987.

\bibitem{alzaid1988first}
A~Alzaid and M~Al-Osh.
\newblock First-order integer-valued autoregressive (inar (1)) process:
  distributional and regression properties.
\newblock {\em Statistica Neerlandica}, 42(1):53--61, 1988.

\bibitem{alzaid1993some}
Abdulhamid~A Alzaid and Mohamed~A Al-Osh.
\newblock Some autoregressive moving average processes with generalized
  {P}oisson marginal distributions.
\newblock {\em Annals of the Institute of Statistical Mathematics},
  45:223--232, 1993.

\bibitem{alzaid2010poisson}
Abdulhamid~A Alzaid and Maha~A Omair.
\newblock On the {P}oisson difference distribution inference and applications.
\newblock {\em Bulletin of the Malaysian Mathematical Sciences Society. Second
  Series}, 33(1):17--45, 2010.

\bibitem{alzaid2012extended}
Abdulhamid~A Alzaid and Maha~A Omair.
\newblock An extended binomial distribution with applications.
\newblock {\em Communications in Statistics-Theory and Methods},
  41(19):3511--3527, 2012.

\bibitem{alzaid2014poisson}
Abdulhamid~A Alzaid and Maha~A Omair.
\newblock Poisson difference integer valued autoregressive model of order one.
\newblock {\em Bulletin of the Malaysian Mathematical Sciences Society},
  37(2):465--485, 2014.

\bibitem{andersson2014parametric}
Jonas Andersson and Dimitris Karlis.
\newblock A parametric time series model with covariates for integers in {Z}.
\newblock {\em Statistical Modelling}, 14(2):135--156, 2014.

\bibitem{billingsley1961statistical}
Patrick Billingsley.
\newblock Statistical methods in {M}arkov chains.
\newblock {\em The annals of mathematical statistics}, pages 12--40, 1961.

\bibitem{billingsley2013convergence}
Patrick Billingsley.
\newblock {\em Convergence of probability measures}.
\newblock John Wiley \& Sons, 2013.

\bibitem{bulla2017bivariate}
Jan Bulla, Christophe Chesneau, and Maher Kachour.
\newblock A bivariate first-order signed integer-valued autoregressive process.
\newblock {\em Communications in Statistics-Theory and Methods},
  46(13):6590--6604, 2017.

\bibitem{chesneau2012parametric}
Christophe Chesneau and Maher Kachour.
\newblock A parametric study for the first-order signed integer-valued
  autoregressive process.
\newblock {\em Journal of Statistical Theory and Practice}, 6:760--782, 2012.

\bibitem{devroye2002simulating}
Luc Devroye.
\newblock Simulating {B}essel random variables.
\newblock {\em Statistics \& probability letters}, 57(3):249--257, 2002.

\bibitem{franke1993conditional}
J~Franke and TH~Seligmann.
\newblock Conditional maximum likelihood estimates for {INAR} (1) processes and
  their application to modelling epileptic seizure counts.
\newblock {\em Developments in time series analysis}, pages 310--330, 1993.

\bibitem{freeland2010true}
R~Keith Freeland.
\newblock True integer value time series.
\newblock {\em AStA Advances in Statistical Analysis}, 94:217--229, 2010.

\bibitem{gupta2006beta}
Arjun~K Gupta and Saralees Nadarajah.
\newblock Beta {B}essel distributions.
\newblock {\em International journal of mathematics and mathematical sciences},
  2006, 2006.

\bibitem{hall2014martingale}
Peter Hall and Christopher~C Heyde.
\newblock {\em Martingale limit theory and its application}.
\newblock Academic press, 2014.

\bibitem{huang2021new}
Jie Huang and Fukang Zhu.
\newblock A new first-order integer-valued autoregressive model with {B}ell
  innovations.
\newblock {\em Entropy}, 23(6):713, 2021.

\bibitem{iliopoulos2003simulation}
George Iliopoulos and Dimitris Karlis.
\newblock Simulation from the {B}essel distribution with applications.
\newblock {\em Journal of Statistical Computation and Simulation},
  73(7):491--506, 2003.

\bibitem{jg1946frequency}
SKELLAM JG.
\newblock The frequency distribution of the difference between two {P}oisson
  variates belonging to different populations.
\newblock {\em Journal of the Royal Statistical Society. Series A (General)},
  109(Pt 3):296--296, 1946.

\bibitem{kachour2011p}
Maher Kachour and Lionel Truquet.
\newblock A p-{O}rder signed integer-valued autoregressive ({SINAR} (p)) model.
\newblock {\em Journal of Time Series Analysis}, 32(3):223--236, 2011.

\bibitem{karlis2006bayesian}
Dimitris Karlis and Ioannis Ntzoufras.
\newblock Bayesian analysis of the differences of count data.
\newblock {\em Statistics in medicine}, 25(11):1885--1905, 2006.

\bibitem{kim2008non}
Hee-Young Kim and Yousung Park.
\newblock A non-stationary integer-valued autoregressive model.
\newblock {\em Statistical papers}, 49:485--502, 2008.

\bibitem{klimko1978conditional}
Lawrence~A Klimko and Paul~I Nelson.
\newblock On conditional least squares estimation for stochastic processes.
\newblock {\em The Annals of statistics}, pages 629--642, 1978.

\bibitem{mckay1932bessel}
AT~McKay.
\newblock A {B}essel function distribution.
\newblock {\em Biometrika}, pages 39--44, 1932.

\bibitem{mckenzie1985some}
Ed~McKenzie.
\newblock Some simple models for discrete variate time series 1.
\newblock {\em JAWRA Journal of the American Water Resources Association},
  21(4):645--650, 1985.

\bibitem{mckenzie1986autoregressive}
Ed~McKenzie.
\newblock Autoregressive moving-average processes with negative-binomial and
  geometric marginal distributions.
\newblock {\em Advances in Applied probability}, 18(3):679--705, 1986.

\bibitem{miletic2018inar}
Ana~V Mileti{\'c}~Ili{\'c}, Miroslav~M Risti{\'c}, Aleksandar~S Nasti{\'c}, and
  Hassan~S Bakouch.
\newblock An {INAR} (1) model based on a mixed dependent and independent
  counting series.
\newblock {\em Journal of Statistical Computation and Simulation},
  88(2):290--304, 2018.

\bibitem{mohammadi2022zero}
Z~Mohammadi, Z~Sajjadnia, HS~Bakouch, and M~Sharafi.
\newblock Zero-and-one inflated {P}oisson--{L}indley {INAR} (1) process for
  modelling count time series with extra zeros and ones.
\newblock {\em Journal of Statistical Computation and Simulation},
  92(10):2018--2040, 2022.

\bibitem{mohammadi2022modeling}
Zohreh Mohammadi, Zahra Sajjadnia, Maryam Sharafi, and Naushad Mamode~Khan.
\newblock Modeling {M}edical {D}ata by {F}lexible {I}nteger-{V}alued {AR} (1)
  {P}rocess with {Z}ero-and-{O}ne-{I}nflated {G}eometric {I}nnovations.
\newblock {\em Iranian Journal of Science and Technology, Transactions A:
  Science}, 46(3):891--906, 2022.

\bibitem{nadarajah2005product}
Saralees Nadarajah and Arjun~K Gupta.
\newblock On the product and ratio of {B}essel random variables.
\newblock {\em International Journal of Mathematics and Mathematical Sciences},
  2005(18):2977--2989, 2005.

\bibitem{nastic2017geometric}
Aleksandar~S Nasti{\'c}, Miroslav~M Risti{\'c}, and Ana V~Mileti{\'c} Ili{\'c}.
\newblock A geometric time-series model with an alternative dependent
  {B}ernoulli counting series.
\newblock {\em Communications in Statistics-Theory and Methods},
  46(2):770--785, 2017.

\bibitem{omair2022bivariate}
Maha~A Omair, Ghadah~A Alomani, and Abdulhamid~A Alzaid.
\newblock Bivariate {D}istributions on {Z}2.
\newblock {\em Bulletin of the Malaysian Mathematical Sciences Society},
  45(Suppl 1):425--444, 2022.

\bibitem{pegram1980autoregressive}
GGS Pegram.
\newblock An autoregressive model for multilag {M}arkov chains.
\newblock {\em Journal of Applied Probability}, 17(2):350--362, 1980.

\bibitem{qi2019modeling}
Xiaohong Qi, Qi~Li, and Fukang Zhu.
\newblock Modeling time series of count with excess zeros and ones based on
  {INAR} (1) model with zero-and-one inflated poisson innovations.
\newblock {\em Journal of Computational and Applied Mathematics}, 346:572--590,
  2019.

\bibitem{shahtahmassebi2016application}
Golnaz Shahtahmassebi and Rana Moyeed.
\newblock An application of the generalized {P}oisson difference distribution
  to the {B}ayesian modelling of football scores.
\newblock {\em Statistica Neerlandica}, 70(3):260--273, 2016.

\bibitem{shirozhan2020inar}
Masoumeh Shirozhan and Mehrnaz Mohammadpour.
\newblock An {INAR} (1) model based on the pegram and thinning operators with
  serially dependent innovation.
\newblock {\em Communications in Statistics-Simulation and Computation},
  49(10):2617--2638, 2020.

\bibitem{tjostheim1986estimation}
Dag Tj{\o}stheim.
\newblock Estimation in nonlinear time series models.
\newblock {\em Stochastic Processes and their Applications}, 21(2):251--273,
  1986.

\bibitem{wang2010generalized}
Dehui Wang and Haixiang Zhang.
\newblock Generalized {RCINAR} (p) process with signed thinning operator.
\newblock {\em Communications in Statistics—Simulation and
  Computation{\textregistered}}, 40(1):13--44, 2010.

\bibitem{weiss2008thinning}
Christian~H Wei{\ss}.
\newblock Thinning operations for modeling time series of counts—a survey.
\newblock {\em AStA Advances in Statistical Analysis}, 92:319--341, 2008.

\bibitem{yang2018integer}
Kai Yang, Dehui Wang, Boting Jia, and Han Li.
\newblock An integer-valued threshold autoregressive process based on negative
  binomial thinning.
\newblock {\em Statistical Papers}, 59:1131--1160, 2018.

\bibitem{yuan2000bessel}
Lin Yuan and John~D Kalbfleisch.
\newblock On the {B}essel distribution and related problems.
\newblock {\em Annals of the Institute of Statistical Mathematics},
  52:438--447, 2000.

\bibitem{zeger1988regression}
Scott~L Zeger.
\newblock A regression model for time series of counts.
\newblock {\em Biometrika}, 75(4):621--629, 1988.

\bibitem{zhang2016properties}
Chi Zhang, Guo-Liang Tian, and Kai-Wang Ng.
\newblock Properties of the zero-and-one inflated {P}oisson distribution and
  likelihood-based inference methods.
\newblock {\em Statistics and its interface}, 9(1):11--32, 2016.

\bibitem{zhang2010inference}
Haixiang Zhang, Dehui Wang, and Fukang Zhu.
\newblock Inference for {INAR} (p) processes with signed generalized power
  series thinning operator.
\newblock {\em Journal of Statistical Planning and Inference}, 140(3):667--683,
  2010.

\bibitem{zheng2007first}
Haitao Zheng, Ishwar~V Basawa, and Somnath Datta.
\newblock First-order random coefficient integer-valued autoregressive
  processes.
\newblock {\em Journal of Statistical Planning and Inference}, 137(1):212--229,
  2007.

\end{thebibliography}
\appendix
\section{Appendix}
\subsection{Proof of theorem 1}
The proof of the theorem follows Theorem 2.1 of Billingsley\cite{billingsley1961statistical}.
We show the conditions(C1)-(C6) are satisfied for the model.
We first give one lemma, which plays a key role in the proofs of other lemmas.
\begin{enumerate}[i.]
	\item $P(m,n)$ is three times continuously differentiable with respect to $\phi,p,\theta,\theta_1,\theta_2$, and for any $m$,$\{n;P(m,n)>0\}$ does not depend on $\omega$. Therefore, $\log P(m,n)$ is well-defined except on a set of $P(m,\cdot)$-measure 0 which does not depend on the parameter values. Thus, conditions (C1) and (C2) are satisfied for the model.
	\item %
Recall the one step transition probability is
\begin{align*}
\notag \pi(z_{t}|z_{t-1}) & = P(Z_{t}=z_{t}|Z_{t-1}=z_{t-1})\\ \notag
& =\phi P(S_{\alpha,\theta}(Z_{t-1})=\delta z_{t})+(1-\phi)P(\varepsilon_{t}=z_{t}) \\ 
& = \phi \frac{I_{\delta z_t}(2 p \beta) I_{z_{t-1}-\delta z_t}(2(1-p) \beta)}{I_{z_{t-1}}(2 \beta)}
+(1-\phi)e^{-\theta_1-\theta_2}\left(\frac{\theta_1}{\theta_2}\right)^{z_t / 2} I_{z_{t}}\left(2 \sqrt{\theta_1 \theta_2}\right)
\end{align*}

Using the derivative formula of modified Bessel function of the first kind:
\begin{align*}
\frac{\partial I_y(\theta)}{\partial \theta}=\frac{y}{\theta} I_y(\theta)+I_{y+1}(\theta).
\end{align*}
Without loss of generality, and in order to simplify the problem, we consider the case $\delta = 1$ (The case $\delta = -1$ is similarly obtained). \\
Let 
$P(m,n) = P(Z_t = n|Z_{t-1} = m) = \phi Z(m,n) + (1-\phi) q(m,n)$,\\
where 
$$
Z(m,n) = \frac{I_{n}(2 p \beta) I_{m-n}(2(1-p) \beta)}{I_{m}(2 \beta)},
$$ 
and 
$$
q(m,n) = e^{-\theta_1-\theta_2}\left(\frac{\theta_1}{\theta_2}\right)^{n / 2} I_{n}\left(2 \sqrt{\theta_1 \theta_2}\right).
$$
Thus  $$ P(m,n) \leq Z(m,n) + q(m,n).$$
And for the first  order derivative of transition probability, we have
$$
\sum_{n = -\infty}^{\infty}\left|\frac{\partial  P(m, n)}{\partial \phi}\right| = \sum_{n = -\infty}^{\infty}\left|Z(m, n)-q(m, n) \right| \leq \sum_{n = -\infty}^{\infty}\left|Z(m, n)\right| + \sum_{n = -\infty}^{\infty}\left|q(m, n)\right| = 2 < \infty ,
$$
the second derivative with respect to $\phi$ is equal to zero.
And the derivative can be simplied as follows:
$$
\left|\frac{\partial P(m, n)}{\partial p}\right| = \left|\frac{\partial Z(m, n)}{\partial p}\right|
,
\left|\frac{\partial  P(m, n)}{\partial \beta}\right| = \left|\frac{\partial  Z(m, n)}{\partial \beta}\right|,
$$
$$
\left|\frac{\partial  P(m, n)}{\partial \theta_1}\right| = \left|\frac{\partial  q(m, n)}{\partial \theta_1}\right|
,
\left|\frac{\partial  P(m, n)}{\partial \theta_2}\right| = \left|\frac{\partial  q(m, n)}{\partial \theta_2}\right|.
$$

Further we have
\begin{align*}
\frac{\partial Z(m,n)}{\partial p} &=-\frac{\beta  \phi  I_n(2 p \beta ) (I_{m-n-1}(2 (1-p) \beta )+I_{m-n+1}(2 (1-p) \beta ))}{I_m(2 \beta )}
 \\
& \quad +\frac{\beta  \phi  (I_{n-1}(2 p \beta )+I_{n+1}(2 p \beta )) I_{m-n}(2 (1-p) \beta )}{I_m(2 \beta )},
\\
\frac{\partial Z(m,n)}{\partial \beta} &= -\frac{\phi  (I_{m-1}(2 \beta )+I_{m+1}(2 \beta )) I_n(2 p \beta ) I_{m-n}(2 (1-p) \beta )}{I_m(2 \beta ){}^2} \\
& \quad +(1-p) \phi  \frac{I_n(2 p \beta ) (I_{m-n-1}(2 (1-p) \beta )+I_{m-n+1}(2 (1-p) \beta ))}{I_m(2 \beta )}  \\
& \quad +\frac{p \phi  (I_{n-1}(2 p \beta )+I_{n+1}(2 p \beta )) I_{m-n}(2 (1-p) \beta )}{I_m(2 \beta )},
\\
\frac{\partial^2 Z(m,n)}{\partial p^2} &= \phi( -\beta  I_n(2 p \beta ) (-\beta  (I_{m-n-2}(2 (1-p) \beta )+I_{m-n}(2 (1-p) \beta ))\\
& \quad -\beta  (I_{m-n}(2 (1-p) \beta )+I_{m-n+2}(2 (1-p) \beta ))) \\
& \quad -2 \beta ^2 (I_{n-1}(2 p \beta )+I_{n+1}(2 p \beta )) (I_{m-n-1}(2 (1-p) \beta )+I_{m-n+1}(2 (1-p) \beta ))\\
& \quad +\beta  (\beta  (I_{n-2}(2 p \beta )+I_n(2 p \beta ))+\beta  (I_n(2 p \beta )+I_{n+2}(2 p \beta ))) I_{m-n}(2 (1-p) \beta ) )
,
\\
\frac{\partial^2 Z(m,n)}{\partial \beta^2} &= \phi  \left(\frac{2 (I_{m-1}(2 \beta )+I_{m+1}(2 \beta )){}^2}{I_m(2 \beta ){}^3}-\frac{I_{m-2}(2 \beta )+2 I_m(2 \beta )+I_{m+2}(2 \beta )}{I_m(2 \beta ){}^2}\right) \\
& \quad \cdot I_n(2 p \beta ) I_{m-n}(2 (1-p) \beta ) -\frac{2 \phi  (I_{m-1}(2 \beta )+I_{m+1}(2 \beta ))}{I_m(2 \beta ){}^2} \\
& \quad \cdot(1-p) I_n(2 p \beta ) (I_{m-n-1}(2 (1-p) \beta )+I_{m-n+1}(2 (1-p) \beta )) \\
& \quad +p (I_{n-1}(2 p \beta )+I_{n+1}(2 p \beta )) I_{m-n}(2 (1-p) \beta )\\
& \quad (1-p) ((1-p) (I_{m-n-2}(2 (1-p) \beta )+I_{m-n}(2 (1-p) \beta ))\\
& \quad +(1-p) (I_{m-n}(2 (1-p) \beta )+I_{m-n+2}(2 (1-p) \beta ))), \\
\frac{\partial^2 Z(m,n)}{\partial p\partial \beta} &= -\frac{\phi  I_n(2 p \beta ) (I_{m-n-1}(2 (1-p) \beta )+I_{m-n+1}(2 (1-p) \beta ))}{I_m(2 \beta )} \\
& \quad +\frac{\beta  \phi  (I_{m-1}(2 \beta )+I_{m+1}(2 \beta )) I_n(2 p \beta ) (I_{m-n-1}(2 (1-p) \beta )+I_{m-n+1}(2 (1-p) \beta ))}{I_m(2 \beta ){}^2}\\ 
& \quad -\frac{\beta  \phi  I_n(2 p \beta ) ((1-p) (I_{m-n-2}(2 (1-p) \beta )+I_{m-n}(2 (1-p) \beta ))}{I_m(2 \beta )} \\ 
& \quad +\frac{(1-p) (I_{m-n}(2 (1-p) \beta )+I_{m-n+2}(2 (1-p) \beta )))}{I_m(2 \beta )} \\
& \quad +\frac{\phi  (I_{n-1}(2 p \beta )+I_{n+1}(2 p \beta )) I_{m-n}(2 (1-p) \beta )}{I_m(2 \beta )} \\
& \quad -\frac{\beta  \phi  (I_{m-1}(2 \beta )+I_{m+1}(2 \beta )) (I_{n-1}(2 p \beta )+I_{n+1}(2 p \beta )) I_{m-n}(2 (1-p) \beta )}{I_m(2 \beta ){}^2} \\
& \quad +\frac{\beta  (1-p) \phi  (I_{n-1}(2 p \beta )+I_{n+1}(2 p \beta )) (I_{m-n-1}(2 (1-p) \beta )+I_{m-n+1}(2 (1-p) \beta ))}{I_m(2 \beta )} \\
& \quad -\frac{\beta  p \phi  (I_{n-1}(2 p \beta )+I_{n+1}(2 p \beta )) (I_{m-n-1}(2 (1-p) \beta )+I_{m-n+1}(2 (1-p) \beta ))}{I_m(2 \beta )} \\
& \quad +\frac{\beta  \phi  (p (I_{n-2}(2 p \beta )+I_n(2 p \beta ))+p (I_n(2 p \beta )+I_{n+2}(2 p \beta ))) I_{m-n}(2 (1-p) \beta )}{I_m(2 \beta )}
.
\end{align*}
Based on the properties of the modified Bessel function of the first kind,
$$
\sum_{y=-\infty}^{\infty}\left(\frac{\theta_1}{\theta_2}\right)^{\frac{y}{2}} I_y\left(2 \sqrt{\theta_1 \theta_2}\right)=e^{\theta_1+\theta_2},
$$
$$
\sum_{y=-\infty}^{\infty} y\left(\frac{\theta_1}{\theta_2}\right)^{\frac{y}{2}} I_y\left(2 \sqrt{\theta_1 \theta_2}\right)=\left(\theta_1-\theta_2\right) e^{\theta_1+\theta_2},
$$
$$
\sum_{y=-\infty}^{\infty} y^2\left(\frac{\theta_1}{\theta_2}\right)^{\frac{y}{2}} I_y\left(2 \sqrt{\theta_1 \theta_2}\right)=\left(\theta_1+\theta_2+\left(\theta_1-\theta_2\right)^2\right) e^{\theta_1+\theta_2},
$$
$$
I_y(\theta) = I_{-y}(\theta),
\sum_{y=-\infty}^{\infty} I_y(\theta)=e^\theta,
\sum_{y=-\infty}^{\infty} y I_y(\theta)=0,
$$
$$
\sum_{y=-\infty}^{\infty} y^2 I_y(\theta)=\theta e^\theta \quad  \forall \theta>0,
\frac{\partial I_y(\theta)}{\partial \theta}=\frac{1}{2}\left(I_{y-1}(\theta)+I_{y+1}(\theta)\right),
$$
thus,
$$
\frac{e^\theta} {2} \leq \sum_{y=0}^{\infty} y I_y(\theta) \leq \frac{\theta e^\theta} {2},
$$
and
$$
e^\theta \leq 2\sum_{y=0}^{\infty} y I_y(\theta) = 2\sum_{y=0}^{\infty} \left|y I_y(\theta)\right|\leq 2\sum_{y=0}^{\infty} y^2 \left|I_y(\theta)\right| \leq
\theta  e^\theta,
$$
Then, for $\frac{\partial Z(m,n)}{\partial p}$, The problem is to estimate the bound of 
$$
\sum_{y=-\infty}^{\infty} \frac{I_y(2p\beta)I_{m-y-1}(2(1-p)\beta)}{I_m(2\beta)},
$$
it is obvious that
$$
\sum_{y=-\infty}^{\infty} \frac{I_y(2p\beta)I_{m-y-1}(2(1-p)\beta)}{I_m(2\beta)} = \frac{1}{I_m(2\beta)}\sum_{y=-\infty}^{\infty}I_y(2p\beta)I_{m-y-1}(2(1-p)\beta) ,
$$
thus we only need to bound
$$
\sum_{y=-\infty}^{\infty}I_y(2p\beta)I_{m-y-1}(2(1-p)\beta) ,
$$
by symmestry,
$$
\sum_{y=-\infty}^{\infty}I_y(2p\beta)I_{m-y-1}(2(1-p)\beta) = \sum_{y=-\infty}^{\infty}I_y(2p\beta)I_{y-m+1}(2(1-p)\beta).
$$
Now we consider the value of $I_y(\theta)$. We know that for  fixed $\theta$,
$$
\sum_{y=0}^{\infty} I_y(\theta)=\frac{e^\theta}{2},
$$
and there is a constant $M$ such that 
$$
\frac{e^\theta}{2} \leq M
,
$$
thus the number of $y$ such that $|I_y(\theta)| \leq 1$ is almost $M-1$. Thus other infinitely many $y$ satisfy $ |I_y(\theta)| < 1 $, moreover, 
$I_y(\theta)^2 \leq |I_y(\theta)|$. Denote 
$
I = \sum_{y=-\infty}^{\infty}I_y(2p\beta)I_{y-m+1}(2(1-p)\beta).
$
We can find large enough constant $M$ such that
\begin{align*}
I&
=
\sum_{y=k_1}^{k_M}I_y(2p\beta)I_{y-m+1}(2(1-p)\beta) 
+ \sum_{y=k_{M+1}}^{\infty}I_y(2p\beta)I_{y-m+1}(2(1-p)\beta) \\
&\leq \sum_{y=k_1}^{k_M}I_y(2p\beta)I_{y-m+1}(2(1-p)\beta) 
 +\frac{1}{2}\sum_{y=k_{M+1}}^{\infty}\left(I_y(2p\beta)^2 + I_{y-m+1}(2(1-p)\beta)^2\right)\\
&\leq \sum_{y=k_1}^{k_M}I_y(2p\beta)I_{y-m+1}(2(1-p)\beta) 
+\frac{1}{2}\sum_{y=k_{M+1}}^{\infty}\left(I_y(2p\beta) + I_{y-m+1}(2(1-p)\beta)\right)
\end{align*}
where $k_1,k_2,\ldots,k_M,k_{M+1},\ldots$ is a rearrangement of $N$ such that when $y \in \{k_1,k_2,\ldots,k_M\}$, $|I_y(\theta)| \geq 1$, 
and when $y \in \{k_M,\ldots,\}$, 
$|I_y(\theta)|<1$.
It can be seen that the former of the right of the inequation is finite, and the latter is convergent.
For higher derivatives, the problem to deal with is the product of  more Bessel functions,
and the rest follows the same pattern.
Thus, we can see that, for fixed $m$,
$$
\sum_{n = -\infty}^{\infty}\left|\frac{\partial Z(m,n)}{\partial p}\right| < \infty
,
\sum_{n = -\infty}^{\infty}\left|\frac{\partial Z(m,n)}{\partial \beta}\right| < \infty,
$$
$$
\sum_{n = -\infty}^{\infty}\left|\frac{\partial^2 Z(m,n)}{\partial p^2}\right| < \infty
,
\sum_{n = -\infty}^{\infty}\left|\frac{\partial^2 Z(m,n)}{\partial \beta^2}\right| < \infty,
\sum_{n = -\infty}^{\infty}\left|\frac{\partial^2 Z(m,n)}{\partial p\partial \beta}\right| < \infty.
$$
Next we are on the position to discuss the derivative of $q(m,n)$.
\begin{align*}
\frac{\partial q(m,n)}{\partial \theta_1} &  = -e^{-\theta_1-\theta_2} \left(\frac{\theta_1}{\theta_2}\right)^{n/2} I_n\left(2 \sqrt{\theta_1 \theta_2}\right) \\
& \quad +\frac{1}{2 \theta_2} n   e^{-\theta_1-\theta_2} \left(\frac{\theta_1}{\theta_2}\right)^{\frac{n}{2}-1} I_n\left(2 \sqrt{\theta_1 \theta_2}\right) \\
& \quad +\frac{1}{2 \sqrt{\theta_1 \theta_2}} \theta_2   e^{-\theta_1-\theta_2} \left(\frac{\theta_1}{\theta_2}\right)^{n/2} \left(I_{n-1}\left(2 \sqrt{\theta_1 \theta_2}\right)+I_{n+1}\left(2 \sqrt{\theta_1 \theta_2}\right)\right) ,\\
\frac{\partial q(m,n)}{\partial \theta_2} &  =   -e^{-\theta_1-\theta_2} \left(\frac{\theta_1}{\theta_2}\right)^{n/2} I_n\left(2 \sqrt{\theta_1 \theta_2}\right)\\
& \quad -\frac{1}{2 \theta_2^2} \theta_1 n   e^{-\theta_1-\theta_2} \left(\frac{\theta_1}{\theta_2}\right)^{\frac{n}{2}-1} I_n\left(2 \sqrt{\theta_1 \theta_2}\right)\\
& \quad +\frac{1}{2 \sqrt{\theta_1 \theta_2}} \theta_1   e^{-\theta_1-\theta_2} \left(\frac{\theta_1}{\theta_2}\right)^{n/2} \left(I_{n-1}\left(2 \sqrt{\theta_1 \theta_2}\right)+I_{n+1}\left(2 \sqrt{\theta_1 \theta_2}\right)\right), \\
\frac{\partial^{2} q(m,n)}{\partial \theta^{2}_1} & = 
\left(\frac{\left(\frac{n}{2}-1\right) n e^{-\theta_1-\theta_2} \left(\frac{\theta_1}{\theta_2}\right)^{\frac{n}{2}-2}}{2 \theta_2^2}+e^{-\theta_1-\theta_2} \left(\frac{\theta_1}{\theta_2}\right)^{n/2}-\frac{n e^{-\theta_1-\theta_2} \left(\frac{\theta_1}{\theta_2}\right)^{\frac{n}{2}-1}}{\theta_2}\right) I_n\left(2 \sqrt{\theta_1 \theta_2}\right) \\
& \quad +  \left(\frac{n e^{-\theta_1-\theta_2} \left(\frac{\theta_1}{\theta_2}\right)^{\frac{n}{2}-1}}{2 \theta_2}-e^{-\theta_1-\theta_2} \left(\frac{\theta_1}{\theta_2}\right)^{n/2}\right) \left(I_{n-1}\left(2 \sqrt{\theta_1 \theta_2}\right)+I_{n+1}\left(2 \sqrt{\theta_1 \theta_2}\right)\right) \frac{\theta_2}{\sqrt{\theta_1 \theta_2}} \\
& \quad + \left( I_{n-2}\left(2 \sqrt{\theta_1 \theta_2}\right)+2I_n\left(2 \sqrt{\theta_1 \theta_2}\right)
+I_{n+2}\left(2 \sqrt{\theta_1 \theta_2}\right) \right) e^{-\theta_1-\theta_2} \left(\frac{\theta_1}{\theta_2}\right)^{n/2} \frac{\theta_2}{4\theta_1}   
 \\
& \quad - \left(I_{n-1}\left(2 \sqrt{\theta_1 \theta_2}\right)+I_{n+1}\left(2 \sqrt{\theta_1 \theta_2}\right)\right) e^{-\theta_1-\theta_2} \left(\frac{\theta_1}{\theta_2}\right)^{n/2} \frac{\theta_2^2}{4 (\theta_1 \theta_2)^{3/2}}, \\
\frac{\partial^{2} q(m,n)}{\partial \theta^{2}_2} & = 
 \left(\frac{\theta_1^2 \left(\frac{n}{2}-1\right) n \left(\frac{\theta_1}{\theta_2}\right)^{\frac{n}{2}-2}}{2 \theta_2^4}+\frac{\theta_1 n \left(\frac{\theta_1}{\theta_2}\right)^{\frac{n}{2}-1}}{\theta_2^3}\right)I_n\left(2 \sqrt{\theta_1 \theta_2}\right) e^{-\theta_1-\theta_2}\\
& \quad +\left(\frac{\theta_1 n e^{-\theta_1-\theta_2} \left(\frac{\theta_1}{\theta_2}\right)^{\frac{n}{2}-1}}{\theta_2^2} 
 +e^{-\theta_1-\theta_2} \left(\frac{\theta_1}{\theta_2}\right)^{n/2} \right)
I_n\left(2 \sqrt{\theta_1 \theta_2}\right) \\
& \quad +  \left(-\frac{\theta_1 n e^{-\theta_1-\theta_2} \left(\frac{\theta_1}{\theta_2}\right)^{\frac{n}{2}-1}}{2 \theta_2^2}-e^{-\theta_1-\theta_2} \left(\frac{\theta_1}{\theta_2}\right)^{n/2}\right)\frac{\theta_1}{\sqrt{\theta_1 \theta_2}} \\
& \quad \quad \cdot \left(I_{n-1}\left(2 \sqrt{\theta_1 \theta_2}\right)+I_{n+1}\left(2 \sqrt{\theta_1 \theta_2}\right)\right)\\
& \quad -  \left(I_{n-1}\left(2 \sqrt{\theta_1 \theta_2}\right)+I_{n+1}\left(2 \sqrt{\theta_1 \theta_2}\right)\right)
e^{-\theta_1-\theta_2} \left(\frac{\theta_1}{\theta_2}\right)^{n/2}\frac{\theta_1^2}{4 (\theta_1 \theta_2)^{3/2}} \\
& \quad +  \left( I_{n-2}\left(2 \sqrt{\theta_1 \theta_2}\right)+2I_n\left(2 \sqrt{\theta_1 \theta_2}\right)
+I_{n+2}\left(2 \sqrt{\theta_1 \theta_2}\right)\right) 
e^{-\theta_1-\theta_2} \left(\frac{\theta_1}{\theta_2}\right)^{n/2}\frac{\theta_1^2}{4 \theta_1 \theta_2},\\
\frac{\partial^2 q(m,n)}{\partial \theta_1 \partial \theta_2} 
& =  -\frac{1}{2 \theta_2}n e^{-\theta_1-\theta_2} \left(\frac{\theta_1}{\theta_2}\right)^{\frac{n}{2}-1} I_n\left(2 \sqrt{\theta_1 \theta_2}\right) +e^{-\theta_1-\theta_2} \left(\frac{\theta_1}{\theta_2}\right)^{n/2} I_n\left(2 \sqrt{\theta_1 \theta_2}\right) \\
& \quad 
-  \left(\frac{n}{2}-1\right) n e^{-\theta_1-\theta_2} \left(\frac{\theta_1}{\theta_2}\right)^{\frac{n}{2}-2}\frac{\theta_1}{2 \theta_2^3} I_n\left(2 \sqrt{\theta_1 \theta_2}\right)\\
& \quad +\frac{(\theta_1-1)}{2 \theta_2^2}
  n e^{-\theta_1-\theta_2} \left(\frac{\theta_1}{\theta_2}\right)^{\frac{n}{2}-1} I_n\left(2 \sqrt{\theta_1 \theta_2}\right) \\
& \quad -\frac{\theta_1 \theta_2}{4 (\theta_1 \theta_2)^{3/2}}  e^{-\theta_1-\theta_2} \left(\frac{\theta_1}{\theta_2}\right)^{n/2} \left(I_{n-1}\left(2 \sqrt{\theta_1 \theta_2}\right)+I_{n+1}\left(2 \sqrt{\theta_1 \theta_2}\right)\right) \\
& \quad +\frac{(1-\theta_1-\theta_2)}{2 \sqrt{\theta_1 \theta_2}}
e^{-\theta_1-\theta_2} \left(\frac{\theta_1}{\theta_2}\right)^{n/2} \left(I_{n-1}\left(2 \sqrt{\theta_1 \theta_2}\right)+I_{n+1}\left(2 \sqrt{\theta_1 \theta_2}\right)\right)\\
& \quad + \frac{1}{4 } \left( I_{n-2}\left(2 \sqrt{\theta_1 \theta_2}\right)+2I_n\left(2 \sqrt{\theta_1 \theta_2}\right)
+I_{n+2}\left(2 \sqrt{\theta_1 \theta_2}\right)\right) 
e^{-\theta_1-\theta_2} \left(\frac{\theta_1}{\theta_2}\right)^{n/2}.
\end{align*}

And we rewrite the derivatives as follows:
$$
\frac{\partial q(m,n)}{\partial \theta_1} = \left(\frac{n}{2\theta_1}-1\right)   q(m,n) + \frac{1}{2} q(m,n-1) + \frac{\theta_2}{2\theta_1} q(m,n+1), 
$$
$$
\frac{\partial q(m,n)}{\partial \theta_2} = \left(-\frac{n}{2\theta_1}-1\right)   q(m,n) + \frac{\theta_1}{2\theta_2} q(m,n-1) +  \frac{1}{2}q(m,n+1), 
$$
\begin{align*}
\frac{\partial^{2} q(m,n)}{\partial \theta_1^{2}} 
&= \left(\frac{n^2-2n}{4\theta_1^{2}} - \frac{n}{\theta_1} + 1\right)   q(m,n) 
+ \left(\frac{2n-1}{4\theta_1 } -1\right) q(m,n-1) \\
 &\quad + \left(\frac{(2n-\theta_1)\theta_2}{4\theta_1^2} - \frac{\theta_2}{\theta_1}\right) q(m,n+1) +\frac{1}{2}q(m,n-2),\\  
\frac{\partial^{2} q(m,n)}{\partial \theta_2^{2}} &= \left(\frac{n^2-2n}{4\theta_2^{2}} + e^{-\theta_1-\theta_2} + \frac{2n}{\theta_2}\right)   q(m,n) 
+ \left( \frac{\theta_1}{4\theta_2^2}-1 - \frac{n\theta_1^2}{2\theta_2^3}    \right) q(m,n-1) \\
& \quad + \left(\frac{1}{4\theta_2} -1 -\frac{n}{2\sqrt{\theta_1\theta_2}} \right)q(m,n+1) 
+\frac{\theta_1^{2}}{2\theta_2^{2}}q(m,n-2), 
\end{align*}

since the modified Bessel function of the first kind is bounded, 
thus  it is obvious that
$$
\left|\sum_{n = -\infty}^{\infty}\frac{\partial q(m,n)}{\partial \theta_1}\right| < \infty,
\left|\sum_{n = -\infty}^{\infty}\frac{\partial q(m,n)}{\partial \theta_2}\right| < \infty,
$$
$$
\left|\sum_{n = -\infty}^{\infty}\frac{\partial^{2} q(m,n)}{\partial \theta_1^{2}}\right| < \infty,
\left|\sum_{n = -\infty}^{\infty}\frac{\partial^{2} q(m,n)}{\partial \theta_1^{2}}\right| < \infty.  
$$
Furthermore, the series converges absolutely, and we obtain
$$
\sum_{n = -\infty}^{\infty}\left|\frac{\partial q(m,n)}{\partial \theta_1}\right| < \infty,
\sum_{n = -\infty}^{\infty}\left|\frac{\partial q(m,n)}{\partial \theta_2}\right| < \infty,
$$
$$
\sum_{n = -\infty}^{\infty}\left|\frac{\partial^{2} q(m,n)}{\partial \theta_1^{2}}\right| < \infty,
\sum_{n = -\infty}^{\infty}\left|\frac{\partial^{2} q(m,n)}{\partial \theta_1^{2}}\right| < \infty.  
$$
 Consequently, the condition (C4)	is satisfied.

Now we caculate the log-partial derivatives for the unknown parameters as follows.\\
For the first  order derivatives, we have
$$
\begin{gathered}
\left|\frac{\partial \log P(m, n)}{\partial \phi}\right|=\left|\frac{Z(m, n)-q(m, n)}{P(m, n)}\right| \leq \frac{Z(m, n)}{P(m, n)}+\frac{q(m, n)}{P(m, n)} 
\leq \frac{1}{\phi(1-\phi)},
\end{gathered}
$$
$$
\begin{aligned}
\left|\frac{\partial \log P(m, n)}{\partial p}\right| = \left|\frac{1}{P(m, n)}\frac{\partial P(m, n)}{\partial p}\right| \leq \psi(m,n) \left|\frac{\partial P(m, n)}{\partial p}\right|,
\end{aligned}
$$
$$
\begin{aligned}
\left|\frac{\partial \log P(m, n)}{\partial \beta}\right| = \left|\frac{1}{P(m, n)}\frac{\partial  P(m, n)}{\partial \beta}\right| \leq \psi(m,n)\left|\frac{\partial  P(m, n)}{\partial \beta}\right|,
\end{aligned}
$$
$$
\left|\frac{\partial \log P(m, n)}{\partial \theta_1}\right| = \left|\frac{1}{P(m, n)}\frac{\partial  P(m, n)}{\partial \theta_1}\right| \leq \psi(m,n)\left|\frac{\partial  P(m, n)}{\partial \theta_1}\right| ,
$$
$$
\left|\frac{\partial \log P(m, n)}{\partial \theta_2}\right| = \left|\frac{1}{P(m, n)}\frac{\partial  P(m, n)}{\partial \theta_2}\right| \leq \psi(m,n) \left|\frac{\partial  P(m, n)}{\partial \theta_2}\right|,
$$
where $\psi$ is a constant relevant to $m$ and $n$,
and combined with the above results, the first derivatives can be bounded by some appropriate constant, respectively.
We know that in the stationary state
$$
E\left|\frac{\partial \log P\left(Z_1, Z_2\right)}{\partial \phi}\right|^2 <\infty,
E\left|\frac{\partial \log P\left(Z_1, Z_2\right)}{\partial p}\right|^2 <\infty,
$$
$$
E\left|\frac{\partial \log P\left(Z_1, Z_2\right)}{\partial \beta}\right|^2 <\infty,
E\left|\frac{\partial \log P\left(Z_1, Z_2\right)}{\partial \theta_1}\right|^2 <\infty,
E\left|\frac{\partial \log P\left(Z_1, Z_2\right)}{\partial \theta_2}\right|^2 <\infty.
$$
Due to the length of the paper, we will not present the second and third order derivatives, the core problem of estimating the bound of derivatives is the method to estimate the bound of product of Bessel functions  when  $m$ is fixed.\\
Analogously, denote $\boldsymbol{\omega}_1 = \phi$,
$\boldsymbol{\omega}_2 = p$,
$\boldsymbol{\omega}_3 = \beta$,
$\boldsymbol{\omega}_4 = \theta_1$,
$\boldsymbol{\omega}_5 = \theta_2$,
$$
E\left|\frac{\partial^2 \log P\left(Z_1, Z_2\right)}{\partial \boldsymbol{\omega}_u \partial\boldsymbol{\omega}_v}\right|^2 <\infty,
E\left|\frac{\partial^3 \log P\left(Z_1, Z_2\right)}{\partial \boldsymbol{\omega}_u \partial\boldsymbol{\omega}_v \partial\boldsymbol{\omega}_w}\right| <\infty,
 \quad u,v,w = 1,2,\ldots,5.
$$
The fisher information matrix $(\sigma_{ij})$ is well-defined and by  (C6) it is nonsingular. We can conclude that the conditions (C1)-(C7) are satisfied. 
 $\hfill\square$

\subsection{Proof of Existence of Stationary Solutions and Ergodic}
In this part we begin to show the existence of stationary solutions for the process $\{Z_n\}$, and prove that the stationary distribution is unique, then we present it is also aperiodic and ergodic. First we introduce the following proposition on the existence of stationary solutions for a Markov chain.
\newtheorem{prop}{Proposition}[subsection]
\begin{prop}
If $\{Z_n\}$ is a weak Feller chain with state-space $Z$ and if for any $\eta>0$ there exists a compact set $C \subseteq Z$ such that $P(x,C^c) < \eta$, for all $z \in Z$, where $P(x,C^c) < \eta$ is the transition probability from $z$ to $C^c$, then $\{Z_n\}$ is bounded in probability and there exists at least one stationary distribution for the chain.
\end{prop}
Now we are on the position to verify whether the  Markov chain satisfies Proposition A.2.1.
\begin{proof}
We choose $C = \{-c,-c+1, \ldots, c-1,c\}$, where $c > 0$ and $c \in \mathbb{N}$. It is obvious that $C$ is a compact set since the number of elements in $C$ is finite. By Markov's inequality,  we have
\begin{equation}
\left\{
\begin{aligned}
\nonumber
&P(Z_{t}\geq c|Z_{t-1}= x) \leq \frac{E[Z_{t}|Z_{t-1}=x]}{c}=\frac{\phi p \delta x+(1-\phi)(\theta_1-\theta_2)}{c}, Z_t  \geq 0,\\
&P(Z_{t}\leq -c|Z_{t-1}= x) \leq -\frac{E[Z_{t}|Z_{t-1}=x]}{c}=-\frac{\phi p \delta x+(1-\phi)(\theta_1-\theta_2)}{c}, Z_t  < 0.\\
\end{aligned}
\right.
\end{equation}
Then, given $\eta>0$, choose enough large $c$ such that $\phi p \delta x+(1-\phi)(\theta_1-\theta_2)<c\eta$.
\end{proof}
We next turn to  prove the uniqueness of a stationary distribution. We shall present $\{Z_n\}$  satisfies Doeblin's condition, i.e., there exists a probability measure $\nu$  with the property that, for some $m \geq 1$, $\eta >0$, and $\gamma>0$, such that
$$\nu(B)>\eta \Rightarrow P^m(x,B)\geq \gamma,
 \quad {\forall}x \in X.$$
Then $\{Z_n\}$ has a unique stationary distribution and is uniformly ergodic. We need to verify whether the above conditions hold for $\{Z_n\}$.
\begin{proof}
Define the measure $\nu$ to have unit point mass at $\{c\}$, where $c \in \mathbb{N}$. We only need to consider Borel sets $B$ with $c \in B$, and it is revealed that we can choose proper  $\eta$ such that $\nu(\{c\}) = \nu(B) > \eta$. Since $P(\{c\})<P(B)$, we have
\begin{align*}
P(Z_{t}\in B|Z_{t-1}=x)&\geq P(Z_{t}=c|Z_{t-1}=x) \\
&=
\phi \frac{I_{\delta c}(2 p \beta) I_{x-\delta c}(2(1-p) \beta)}{I_{x}(2 \beta)}
+(1-\phi)e^{-\theta_1-\theta_2}\left(\frac{\theta_1}{\theta_2}\right)^{c / 2} I_{c}\left(2 \sqrt{\theta_1 \theta_2}\right).
\end{align*}
Case 1: If $\phi = 0$, we can choose $\gamma$ such that
\begin{align*}
P(Z_{t}\in B|Z_{t-1}=x)&\geq P(Z_{t}=c|Z_{t-1}=x) =
e^{-\theta_1-\theta_2}\left(\frac{\theta_1}{\theta_2}\right)^{c / 2} I_{c}\left(2 \sqrt{\theta_1 \theta_2}\right)=\gamma>0.
\end{align*}
Case 2: If $\phi = 1$, we can choose $\gamma$ such that
\begin{align*}
P(Z_{t}\in B|Z_{t-1}=x)&\geq P(Z_{t}=c|Z_{t-1}=x) =
\frac{I_{\delta c}(2 p \beta) I_{x-\delta c}(2(1-p) \beta)}{I_{x}(2 \beta)}
=\gamma>0.
\end{align*}
Case 3: If $0<\phi <1$, we can choose $\gamma$ such that
\begin{align*}
P(Z_{t}\in B|Z_{t-1}=x)&\geq P(Z_{t}=c|Z_{t-1}=x) \\&=
\phi \frac{I_{\delta c}(2 p \beta) I_{x-\delta c}(2(1-p) \beta)}{I_{x}(2 \beta)}
+(1-\phi)e^{-\theta_1-\theta_2}\left(\frac{\theta_1}{\theta_2}\right)^{c / 2} I_{c}\left(2 \sqrt{\theta_1 \theta_2}\right).\\
& =\gamma >0.
\end{align*}
Thus we can choose $m = 1$, and Doeblin's condition is satisfied. Denote $P(Z_{t}=x|Z_{t-1}=x)=P_{xx}$, for  and we can calculate
%
\begin{equation}
P_{xx}=\left\{
\begin{aligned}
\nonumber
& \phi \frac{I_{ x}(2 p \beta) I_{0}(2(1-p) \beta)}{I_{x}(2 \beta)}
+(1-\phi)e^{-\theta_1-\theta_2}\left(\frac{\theta_1}{\theta_2}\right)^{x / 2} I_{x}\left(2 \sqrt{\theta_1 \theta_2}\right)>0, \delta  =1,\\
&\phi \frac{I_{ x}(2 p \beta) I_{2 x}(2(1-p) \beta)}{I_{x}(2 \beta)}
+(1-\phi)e^{-\theta_1-\theta_2}\left(\frac{\theta_1}{\theta_2}\right)^{x / 2} I_{x}\left(2 \sqrt{\theta_1 \theta_2}\right)>0, \delta  =-1.\\
\end{aligned}
\right.
\end{equation}
Therefore, $\{Z_n\}$ is aperiodic, and we can conclude that $\{Z_n\}$ is uniformly ergodic.
\end{proof}
\end{enumerate}

\end{document}